\documentclass{article}
\usepackage{dannychen}
\usepackage{wrapfig}
\hypersetup{breaklinks=true, colorlinks=true, linkcolor={blue!90!black}, citecolor={blue!90!black}, urlcolor={blue!90!black}}

\title{Inferring Quantum Network Topology using Local Measurements}
\author[1]{Daniel T. Chen}
\author[2]{Brian Doolittle}
\author[1]{Jeffrey Larson}
\author[1]{Zain H. Saleem}
\author[3]{Eric Chitambar}
\affil[1]{Mathematics and Computer Science Division, Argonne National Laboratory}
\affil[2]{Department of Physics, University of Illinois at Urbana-Champaign}
\affil[3]{Department of Electrical and Computer Engineering, University of Illinois at Urbana-Champaign}

\usepackage{marginnote}

\reversemarginpar

\usepackage{blindtext}

\usepackage{tikz}
\usetikzlibrary{quantikz}

\usepackage[nameinlink]{cleveref}
\crefname{lemma}{lemma}{lemmas}
\Crefname{lemma}{Lemma}{Lemmas}
\crefname{assumption}{assumption}{assumptions}
\Crefname{assumption}{Assumption}{Assumptions}

\newcommand{\subassref}[2]{{\Cref{ass:#1}.\ref{subass:#2}}}
\usepackage{enumitem}
\usepackage{tikz}
\usepackage{tikz-network}

\newcommand{\mbb}{\mathbb}
\newcommand{\mc}{\mathcal}
\newcommand{\ip}[1]{\left\langle #1 \right\rangle} %
\newcommand{\op}[2]{\left|#1\right\rangle\!\left\langle #2\right|} %

\newcommand{\Net}{\text{Net}}

\newcommand{\rhoNet}{\rho^{\Net}}

\newcommand{\av}{\vec{a}}

\newcommand{\thetav}{\vec{\theta}}

\newtheorem{assumption}{Assumption}

\definecolor{cool_green}{rgb}{0.0, 0.5, 0.0}
\newcommand{\eric}[1]{{\color{cool_green} #1}}

\date{}

\begin{document}
\maketitle

\begin{abstract}
    Statistical correlations that can be generated across the nodes in a quantum network depend crucially on its topology. However, this topological information might not be known a priori, or it may need to be verified. In this paper, we propose an efficient protocol for distinguishing and inferring the topology of a quantum network. We leverage entropic quantities---namely, the von Neumann entropy and the measured mutual information---as well as measurement covariance to uniquely characterize the topology. We show that the entropic quantities are sufficient to distinguish two networks that prepare GHZ states. Moreover, if qubit measurements are available, both entropic quantities and covariance can be used to infer the network topology without state-preparation assumptions. We show that the protocol can be entirely robust to noise and can be implemented via quantum variational optimization. Numerical experiments on both classical simulators and quantum hardware show that covariance is generally more reliable for accurately and efficiently inferring the topology, whereas entropy-based methods are often better at identifying the absence of entanglement in the low-shot regime.
\end{abstract}

\section{Introduction}

Quantum entanglement~\cite{Horodecki2009_entanglement_review} is a static resource that can be shared between parties and used to generate correlations.
In nature, entanglement leads to interesting physical phenomena in quantum many-body systems~\cite{Amico2008_manybody_entanglement}, while from an engineering perspective, quantum entanglement offers operational advantages in communication, cryptography, computation, and sensing technologies~\cite{Dowling2003_quantum_tech_2}.

We focus on quantum communication networks in which entanglement can be generated and distributed to different parties~\cite{Kimble2008quantum_internet,Simon2017global_quantum_network,Wehner2018quantum_internet}. %
The entanglement can then be used to assist in a variety of networking applications~\cite{Wehner2018quantum_internet}.
The first quantum networks have been developed and will continue to scale~\cite{Sun2016teleportation,Hermans2022_delft_network_teleportation,Carvacho2022_urban_network_non-n-locality,Chung2022_illinois_express_quantum_metro_network}; therefore, to harness the power of quantum networks, it is important to characterize different types of entanglement structures.

The functionality of a quantum network will depend crucially on its \textit{topology}, which is the particular connectivity structure between the sources and the measurement devices.
While significant effort has been dedicated to detecting multipartite entanglement~\cite{Terhal2002detecting,Guhne2009entanglement_detection,Huber2010_high-dim_multipartite_entanglement,Harney2020_entanglement_classifying_neural_net,Luo2021_robust_multipartite_entanglement,Hansenne2022_quantum_network_symmetries,Scala2022_variational_learning_entanglement}, less is known about inferring a network's topology.
In practice, we need an efficient procedure for inferring network topology from the experimental data.
Furthermore, such a procedure must be compatible with existing quantum systems, meaning it does not rely upon quantum memory or complex multi-qubit measurements.

Recently, several techniques have been developed to assist in this task.
In most approaches, classical data is sampled from a network and tested for compatibility with a given network topology.
Numerous network compatibility tests have been developed, including violations of entropic bounds~\cite{Henson2014_entropic_bounds,chaves2015_entroptic_bounds,Weilenmann2017_entropic_bounds},  network Bell inequalities~\cite{Chaves2016_polynomial,Tavakoli2016tree,Rosset2016_nonlinear_bell_inequalities,Tavakoli2022_network_nonlocality}, and quantum Finner inequalities~\cite{Renou2019_correlation_limits,Luo2021network_configuration}, as well as semidefinite tests on covariance matrices~\cite{Kela2020_semidefinite_tests,Aberg2020_semidefinite_tests,Kraft2021_characterizing_quantum_networks} and inflation techniques~\cite{Wolfe2019_inflation,Wolfe2021_inflation}.

In a separate approach, the von Neumann entropy of independent measurement devices was used to infer whether two networks of GHZ states are equivalent under local unitary transformations~\cite{yang2022strong}.
{Since each qubit in a GHZ state has a maximally mixed reduced density matrix, an analysis of the local von Neumann entropy alone will fail to distinguish between two networks whose measurement nodes receive the same number of qubits (see  \Cref{fig:tri-net-2}).}
As a solution, Yang et al.~\cite{yang2022strong} proposed using the multipartite Shannon mutual information in addition to the von Neumann entropy, however, this solution might be challenging to scale to networks with {$n \gg 0$} nodes {since it requires computing the joint entropy of all $2^n$ subsets of parties.}

\begin{figure}
\begin{center}
        \includegraphics[width=.6\linewidth]{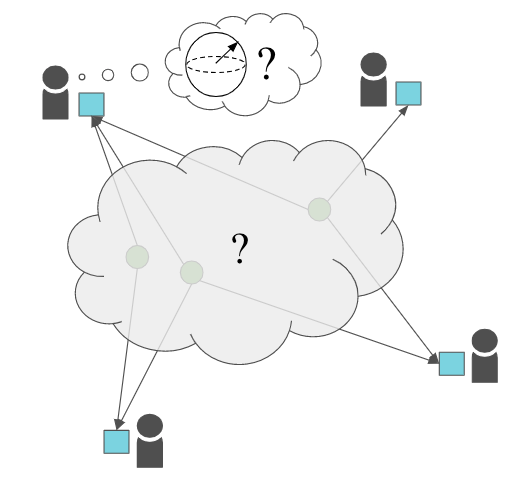}
\end{center}
    \caption{We consider the problem of inferring the topology (covered by the grey cloud) of a quantum network using measurement data on qubits received {from the sources.  In addition to not knowing the exact connectivity or whether all links are functioning properly, the senders and receivers may have misaligned Bloch spheres used for qubit encoding/decoding.} %
    Lastly, we also consider noise corruption when sending qubits via connecting links.}
    \label{fig:main-q}
\end{figure}

The main problem we consider is depicted in \Cref{fig:main-q}.  The network is assumed to have a two-layered structure, the first layer consisting of quantum sources, and the second layer consisting of measurement nodes.  However, the particular connectivity between these two layers is unknown and our goal is to identify this structure by studying the measurement statistics at the nodes.  Furthermore, we do not assume that all nodes on the network share the same reference frame for the encoded information, which means that arbitrary local unitary rotations are permitted at either the source or measurement layers.  As a final complication, we also consider the scenario of noisy connecting links.

We demonstrate two approaches to solving the above problem.
First, we show that networks of different topologies can be distinguished using only local measurements, albeit requiring trustworthy sources to prepare GHZ states.
Building upon Yang et al.~\cite{yang2022strong}, we use the von Neumann entropy to count the number of sources linked to each node, while we use the measured mutual information---the maximum mutual information observed between two parties using measurements local to each---between two nodes to count the number of sources they share.
We extend our analysis to the case where qubit-wise uniform depolarizing noise is applied to the sources. We show that the protocol can be useful for comparing network topologies when the noise strength is known, but limited in utility otherwise.
On the other hand, when each node measures one qubit, we show that a network's topology can be fully characterized also with measurements local to each node, and in a manner robust to noise (we place no assumptions on the noise model). 
Furthermore, our approach is practical because it scales quadratically with the number of qubits, uses only qubit measurements, and does not require a quantum memory.
We also demonstrate empirically a topology inference algorithm that can be implemented on quantum hardware using variational quantum optimization methods~\cite{Doolittle2023}.
The variational optimization approach improves inference capabilities by maximizing the observed correlations while allowing numerical estimations of the Von Neumann entropy and measured mutual information, which are entropic quantities that are otherwise expensive to compute.
We conduct numerical experiments on simulators and quantum hardware to contrast the choice of correlation measures, namely, mutual information and covariance. In general, covariance-based protocols are better at detecting entanglement structures while optimizing more efficiently. However, we also find that entropic methods tend to more reliably identify the absence of entanglement, particularly in the low-shot regime.

The organization of the paper is as follows. In \Cref{sec:prelim}, we formally introduce quantum networks and their relevant entropic quantities, while establishing the notation that we use throughout the manuscript. \Cref{sec:distinguish-network} introduces protocols for distinguishing the topology of two quantum networks. With state preparation assumptions, we show that two networks can be distinguished by observing entropic quantities admitted by individual/pairs of measurement nodes. We also discuss instances where the protocol breaks down under noisy channels. Then, \Cref{sec:variational-inference} briefly discusses variational optimization algorithms for estimating entropic quantities. Lastly, \Cref{sec:infer-network} extends our previous protocol to measurements on individual qubits, which gives a polynomial-time algorithm for inferring the topology of quantum networks. This algorithm does not depend on a priori knowledge of the prepared states and is robust to noise. We implement and test the algorithm on both simulator and quantum hardware to study the practical performance with respect to channel and statistical shot noise.

\section{Preliminary on quantum networks}\label{sec:prelim}

This section introduces formally $n$-local quantum networks. These networks contain several components (sources, nodes, and links) as well as important measurable quantities that can be used to infer the topology of the network.

\subsection{$n$-local quantum networks}
An $n$-local quantum network, the object of interest in this paper, is an $N_q$-qubit system, where $N_q$ denotes the number of qubits. 
Each qubit is indexed by an integer $k \in \{1, 2, \dots, N_q\} \equiv [N_q]$.
A quantum network is characterized by sources $\Lambda_i$, measurement nodes $A_j$, and links $L_k$.
Furthermore, we let $N$ denote the number of a quantity, for example, $N_s$ denotes the number of sources and $N_m$ the number of measurement nodes.

More specifically, we can concisely interpret quantum networks as a directed bipartite graph $G = (\{\bm \Lambda,\bm A\},\bm L)$. 
The vertices are partitioned into the sources $\bm \Lambda = \{ \Lambda_i\}_{i=1}^{N_s}$ and measurement nodes $\bm A = \{ A_j\}_{j=1}^{N_m}$. 
The edges connect sources to nodes, $\bm L = \{(\Lambda_i, A_j)\}$, and represent the movement of qubits. 
See \Cref{fig:network-diagram} for an example quantum network and an enumeration of its respective parts.

\begin{figure}
    \centering
    \includegraphics[width=\linewidth]{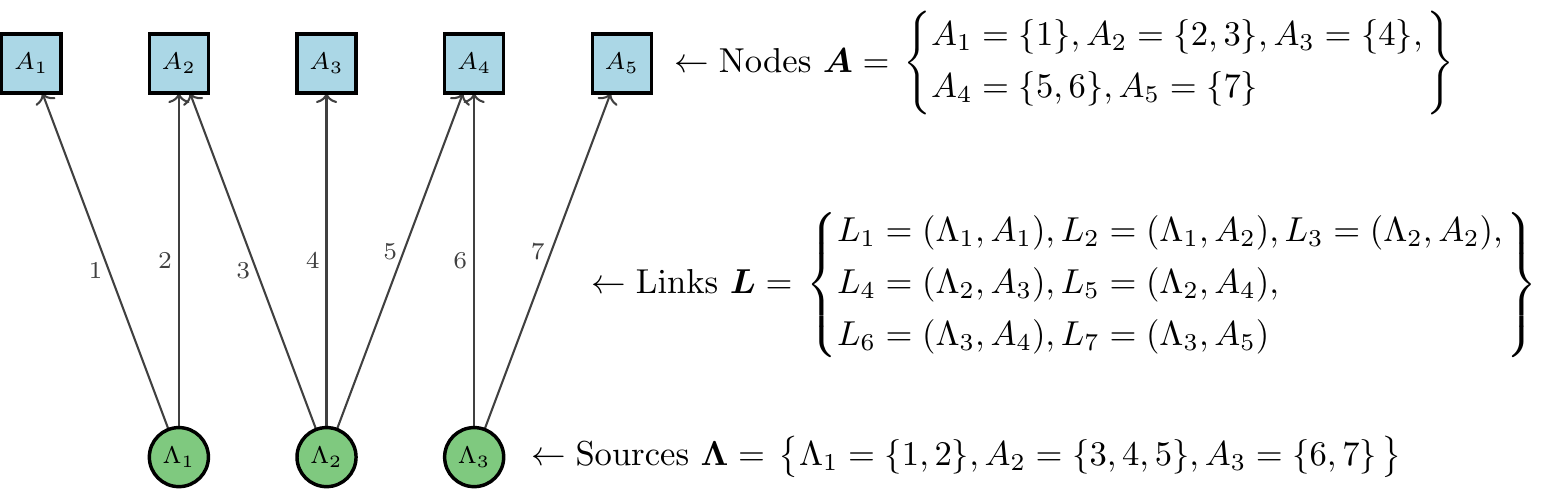}
    \caption{A quantum network is composed of sources (green circles), links (edges), and measurement nodes (blue squares). Each link sends one qubit from a source to a node. Viewing the nodes and sources jointly as the vertex set, a quantum network can be interpreted as a bipartite graph.}
    \label{fig:network-diagram}
\end{figure}

\paragraph{Sources} 
A source, indexed by an integer $i \in [N_s]$, is characterized by the subset of qubits it acts on, $\Lambda_i \subseteq [N_q]$. In a quantum network, $N_s$ sources collectively prepare a state $\ket{\psi} = \bigotimes_{i=1}^{N_s} \ket{\psi^{\Lambda_i}}$ where $\mathcal H$ is a $2^{N_q}$-dimensional Hilbert space and $\ket{\psi}\in \mc{H}$.
There are two frequently used states in the remainder of the manuscript (particularly in \Cref{sec:distinguish-network}). The first is the Greenberger-Horne-Zeilinger (GHZ) state, which takes the form
\begin{align}\label{eq:ghz}
    \ket{\text{GHZ}_n} = \frac{1}{\sqrt{2}} \left( \ket{00\dots 0} + \ket{11\dots1} \right)
\end{align}
where $n$ denotes the number of qubits. When $n = 2$, GHZ states are equivalent to Bell states, denoted by $\ket{\Phi}$. Moreover, partial traces of the GHZ state results in a shared classical random bit, denoted by $\sigma_n$, defined as
\begin{align}
    \sigma_n = (\ket{00\dots0}\bra{00\dots0} + \ket{11\dots1}\bra{11\dots1})/2.
\end{align}
This shared random bit is the other frequently used state. The density matrices of these two states differ by their off-diagonal entries.

\paragraph{Links}
Links $L_k$, for $k \in [N_q]$, 
can be represented graph-theoretically as an edge that connects a source node to a measurement node, $L_k = (\Lambda_i, A_j)$. 
Each link transmits exactly one qubit; hence, there are as many links as qubits.
Physically, they are modeled as a quantum, possibly noisy, channel.
A quantum channel is mathematically defined as a completely-positive trace-preserving (CPTP) map~\cite{Nielsen2009} $\mathcal E : D(\mc{H}) \to D(\mc{H})$, where $D(\mc H)$ denotes the space of density matrices of states in $\mc H$. 
Alternatively, the channel can be expressed in the operator-sum representation~\cite{Nielsen2009, Kraus1983}
\begin{equation}\label{eq:kraus_op_rep}
    \mathcal E(\rho) = \sum_i K_i \rho K_i^{\dagger}, \; \text{where}\;\sum_i K_i^{\dagger}K_i = \mbb{I},
\end{equation}
where $\{K_i\}$ are Kraus operators~\cite{Kraus1983}.
For example, the depolarizing channel for a one-qubit system $\rho$ with strength $\gamma$ has the following Kraus operators,
\begin{align}
    K_0 = \sqrt{1 - \frac{3\gamma}{4}} \mathbb I_2,\, &K_1 = \sqrt{\frac{\gamma}{4}}\sigma_x,\, \nonumber \\ 
    &K_2 = \sqrt{\frac{\gamma}{4}}\sigma_y,\, K_3 = \sqrt{\frac{\gamma}{4}}\sigma_z.
\end{align}
where $\mathbb I_2$ is a $2 \times 2$ identity matrix and $\sigma_y, \sigma_y, \sigma_z$ are the Pauli matrices.
Note, here we are assuming that the channel noise in each link acts independently of each other.

\paragraph{Measurement nodes}
Measurement nodes receive the incoming qubits and output the corresponding measurement outcomes. 
For a node $A_j \subseteq [N_q]$, $j \in [N_m]$, we consider a projection-valued measure (PVM) $\{\Pi^{A_j}_{a_j}\}$ that forms a set of orthogonal projectors satisfying $\sum_{a_j} \Pi^{A_j}_{a_j} = \mbb{I}_{2^{|A_j|}}$.
The node measures its local qubits $\rho^{A_j}\in D(\mc{H}^{A_j})$ that were received from its linked sources.
We assume measurement nodes are independent of one another, and the network applies the projector $\Pi_{\av} = \bigotimes_{j=1}^{N_m} \Pi^{A_j}_{a_j}$.
Upon measurement, the classical output $\av$ is obtained with probability,
\begin{equation}\label{eq:quantum-conditional-probabilities-born-rule}
    \mathbb P(\av) = \tr \left( \Pi_{\av} \,\mathcal E_{\text{tot}} (\rho_{\text{tot}}) \right),
\end{equation}
where $\rho_{\text{tot}}$ designates the total state generated by all the sources and $\mathcal{E}_{\text{tot}}$ is the joint channel across all edges. It is worth noting that any permutations needed to map the joint Hilbert space of sources to that of the measurement nodes are included implicitly.

\subsection{Entropic quantities on quantum networks}
The paper focuses on two entropic quantities observed on networks: the von Neumann entropy and the measured mutual information. 
Both quantities convey important information about the topology of the network and will be discussed in more detail in the next section. 

\paragraph{Von Neumann entropy} The von Neumann entropy for a quantum state $\rho$ is defined as
\begin{align}
    S(\rho) = \mbox{\eric{$-$}}\tr \left( \rho \log \rho \right)
\end{align}
where the $\log(\cdot)$ above refers to the matrix logarithm and we use the convention that $\log 0 = 0$. 
So, any pure state $\rho = \ket \psi \bra \psi$ has a von Neumann entropy of zero. Recall that the Shannon entropy of a probability distribution $\mu$ on support $\mathcal X$ is defined as 
\begin{align}
    H(\mu) = - \sum_{x \in \mathcal X} \mu(x) \log \mu(x).
\end{align}
When measured in the eigenbasis of $\rho$, the von Neumann entropy coincides with the Shannon entropy of the distribution over measurement outcomes~\cite{Nielsen2009}, with all randomness coming from the superposition of pure states in $\rho$~\cite{Nielsen2009}. 
When measured in any other basis, the Shannon entropy calculated from measurement results is strictly greater than the von Neumann entropy because measurements only add noise. 
Thus, the von Neumann entropy can be calculated by minimizing the Shannon entropy over the measurement basis, i.e.,
\begin{align}\label{eq:vn_entropy_as_min_shannon_entropy}
    S(\rho) = \min_{\{\Pi_{\vec a}\}} H(\mathbb P(\vec a))
\end{align}
where $\{\Pi_{\vec a}\}$ is a complete set of projections and $\mathbb P(\vec a)$ is the probability distribution upon measuring the quantum state in basis $\{\Pi_{\vec a}\}$.

\paragraph{Measured mutual information} Intuitively, the mutual information between two random variables quantifies the amount of correlation between them. However, the conventional mutual information defined for quantum systems involves joint measurement between the two parties. Let $A_i$ and $A_j$ be two measurement devices. We introduce the \textit{measured mutual information} as the maximal mutual information between local measurements distributions generated by $A_i$ and $A_j$,
\begin{equation}\label{eq:meas-mut-info}
\footnotesize
\begin{aligned}
    \mathcal I_m (A_i; A_j) = \max_{\{\Pi^{A_i}_{\vec a_i} \tensor \Pi^{A_j}_{\vec a_j}\}} &\left[H(\mathbb P (\vec a_i)) + H(\mathbb P (\vec a_j)) \right.\\&\left.- H(\mathbb P (\vec a_i, \vec a_j))\right]. 
\end{aligned}
\end{equation}

If the two measurement nodes are not correlated, then we can decompose the joint distribution into products. Furthermore, since the Shannon entropy of independent random variables is additive, the measured mutual information will go to zero if no correlation---quantum or classical---is shared. 

\section{Distinguishing network topology}\label{sec:distinguish-network}

We first are interested in protocols that can distinguish two network topologies. We define two networks to be the same if they are related by a graph isomorphism, formally defined below.

\begin{definition}[Network Topology] \label{def:isomorphism}
Two quantum networks, $\mathcal N \ind 1$ and $\mathcal N \ind 2$, have the same topology if there exists bijections $\phi: [N_s] \to [N_s], \varphi: [N_m] \to [N_m]$ such that for any edge $L_k \ind 1 = (\Lambda_i \ind 1, A_j \ind 1)$, there is a corresponding $L_k \ind 2 = (\Lambda_{i} \ind 2, A_{j} \ind 2) = (\Lambda_{\phi(i)} \ind 1, A_{\varphi(j)} \ind 1)$.
\end{definition}

Note that in the above definition, we gave two bijections $\phi$ and $\varphi$ separately for sources and measurement nodes. 
Conventionally, one bijective map is sufficient to describe the relabeling of vertices. 
In the context of quantum networks, the two maps are necessary to ensure sources and measurement nodes remain distinct.

Yang et al.~gave a protocol for distinguishing the topology of quantum networks using GHZ states~\cite{yang2022strong}. More specifically, they proved that the von Neumann entropies measured at each node are the same between two networks, up to a permutation of node indices, if and only if the topologies of the two networks are the same. However, this theorem only holds if no two measurement nodes share more than one source. We find this class of network restricting. In this section, we will introduce an alternative protocol that removes this restriction and distinguishes the topology of two networks where any pairs of nodes can share any number of sources.

\subsection{Topology classification using von Neumann entropy}\label{sec:topo-classification}
We begin by reviewing one of the results from Yang et al.~\cite{yang2022strong} that we will be extending. Consider an $n$-local quantum network $\mathcal N$ where no two nodes share more than one source. They constructed the \textit{characteristic vector} 
\begin{align}
    V_{\mathcal N} = \begin{pmatrix} S(A_1) & S(A_2) & \dots & S(A_{N_m}) \end{pmatrix}
\end{align}
to store the von Neumann entropy measured on each node. Then, a quantum network can be uniquely characterized by its characteristic vector.

\begin{lemma}[Theorem 6 of~\cite{yang2022strong}]\label{thm:Yang-thm}
Let $\mathcal N \ind 1$ and $\mathcal N \ind 2$ be two quantum networks preparing GHZ states and for any two parties $A_i$ and $A_j$, they share no more than one source, that is, 
\begin{align}
\footnotesize
    \biggr \vert \bigr \{\Lambda_k \in \bm \Lambda: (\Lambda_k, A_i), (\Lambda_k, A_j) \in \bm L,  i \neq j \bigr \} \biggr \vert \leq 1.
\end{align}
Then, $\mathcal N \ind 1$ and $\mathcal N \ind 2$ have the same topology if and only if their characteristic vectors are equal to each other.
\end{lemma}

It would be helpful to establish a graphical interpretation of the von Neumann entropy. 
We want to first establish assumptions on the class of network to be considered for the remainder of the section, enumerated below.

\begin{assumption}\label{ass:all}
We assume:
\begin{enumerate}[label=\textnormal{\textbf{\Alph*}.},ref=\Alph*]
    \item Each source prepares maximally entangled states (GHZ states) up to local unitary transformations. Without loss of generality, assume states are prepared in the form specified in \Cref{eq:ghz} (see remark below). \label{subass:A1}
    \item Each source sends at most one qubit to any given measurement device.\label{subass:A2}
    \item Only measurements local to each measurement device can be performed.\label{subass:A3}
\end{enumerate}
\end{assumption}

\begin{remark}
\Cref{subass:A2} of \Cref{ass:all} was primarily to avoid ``parallel edges.'' Note that the statistics obtained from multiple maximally entangled qubits are identical to those obtained with just one qubit. Therefore, along with the assumption that each link represents the movement of a single qubit, having parallel edges will generate the same statistics at the measurement node and multiple networks can reproduce the same measurement outcome. However, if we allow multiple qubits from the same source to be transmitted in one link, we can relax \Cref{subass:A2} of \Cref{ass:all} at no cost to the protocol's correctness.
\end{remark}

\begin{remark}
As mentioned previously, the basis of choice at the sources can be different than the choice at measurement nodes. However, since calculating the von Neumann entropy and the measured mutual information requires optimization over basis sets at the end of each measurement node, the optimal basis will match the reference frame of the sources. Implementation of such a procedure can be achieved via differential programming (cf.~\Cref{sec:variational-inference} or  \cite{Doolittle2023}). Thus, for the remainder of the manuscript, the term ``computational basis'' would be used synonymously with ``the source's reference frame'' without loss of generality.
\end{remark}

Following \subassref{all}{A2}, we can interpret the von Neumann entropy as a graph-theoretic quantity.

\begin{lemma} \label{thm:VNE-interpretation}
Let a quantum network satisfy \Cref{ass:all}. Then, for any node $A_i$, $S(A_i) = N_s^{A_i}$, where $N_s^{A_i}$ denotes the number of sources $A_i$ is connected to.
\end{lemma}

\begin{proof}
Since only GHZ states are prepared and each source can send at most one qubit, the qubits received at node $A_i$ are all maximally entangled with another qubit that is not present in $A_i$. Thus, the state at node $A_i$ is maximally mixed and has a von Neumann entropy (which, in this case, is equivalent to the Shannon entropy) of $A_i$ is the number of qubits received, $N_s^{A_i}$. 
\end{proof}

Thus, in light of the graph-theoretic interpretation, \Cref{thm:Yang-thm} states that knowing the number of sources connected to each node is sufficient for knowing the topology of a quantum network. However, the assumption that nodes share no more than one entanglement is crucial and restrictive, as emphasized in the following example.

\paragraph{Example: triangle networks}

\begin{figure}
  \centering
    \begin{subfigure}{0.35\linewidth}
    \includegraphics[width=\linewidth]{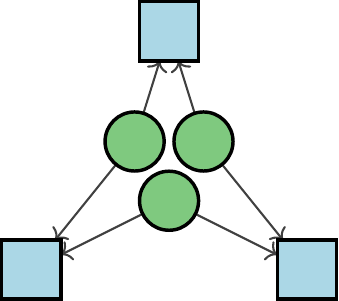}
    \caption{Network 1}
    \label{fig:tri-net-1}
    \end{subfigure}\hspace{30pt}
    \begin{subfigure}{0.35\linewidth}
    \includegraphics[width=\linewidth]{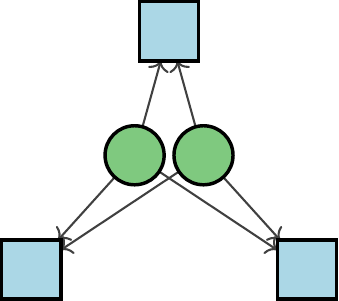}
    \caption{Network 2}
    \label{fig:tri-net-2}
    \end{subfigure}
    \caption{Example triangle networks that are indistinguishable solely from von Neumann entropy.}
    \label{fig:network-example}
\end{figure}

Consider the two networks as shown in \Cref{fig:network-example}. The first network satisfies the assumption Yang et al.~\cite{yang2022strong} made. Each node receives two qubits, each from two different sources. Since the subsystem of any maximally entangled state is a maximally mixed one, the von Neumann entropy at each node is $2$. On the other hand, the second network consists of only two preparation nodes, each preparing a $3$-qubit GHZ state. Each node in network 2 also receives two qubits, one from each source. Again by property of maximally entangled states, the von Neumann entropy at each node is $2$. One could take a step further and study the von Neumann entropy of the joint state of two measurement nodes only to find out that the two networks yield the same statistics. Thus, observing the von Neumann entropy alone cannot distinguish networks. 

As a solution, Yang et al.~\cite{yang2022strong} show that the Shannon mutual information can distinguish between the two networks in \Cref{fig:network-example}. Although the Shannon mutual information is evaluated from classical data, this entropic quantity must be evaluated for all groupings of parties where the number of groupings scales exponentially with the number of parties. Thus, this approach is not practical for large networks.

We propose the addition of the \textit{measured mutual information}. We claim that the basis that maximizes the Shannon entropy is the computational basis, which will be formally proven later. Take any two nodes in network 1. The Shannon entropy at each node will be $2$ since a maximally mixed state is information-theoretically equivalent to a fair classical coin flip. The joint state of the nodes can be written as 
\begin{align}
    \frac{1}{4} \left( \mathbb I_2 \tensor \ket\Phi \bra\Phi \tensor \mathbb I_2 \right),
\end{align} 
which acts equivalently as three independent fair coin flips. This yields a joint Shannon entropy of $3$ with measurements local to each node. Thus, the measured mutual information will be $1$ for all pairs of nodes in network $1$. 

The same does not hold for network 2! The joint state of any pair of nodes in network 2 
\begin{align}
    \frac{1}{4} \left( \ket{00}\bra{00} + \ket{11}\bra{11} \right)^{\tensor 2} 
\end{align}
has a Shannon entropy of $2$ upon measuring separately in the respective nodes. Thus, the measured mutual information in network 2 is $2$.

\subsection{Protocol for distinguishing network topology}

The example above gave evidence for a graph-theoretic interpretation of both entropic quantities---the von Neumann entropy of a node gives the number of sources the node is connected to, whereas the measured mutual information gives the number of sources the two nodes share. We formally present this in the lemma below, whose proof is deferred to \Cref{sec:pf-mmi-interpretation}.

\begin{lemma} \label{thm:MMI-interpretation}
Let a quantum network satisfy \Cref{ass:all}. Then, for any two measurement nodes $A_i$ and $A_j$, $\mathcal I_m (A_i;A_j) = N_s^{A_i,A_j}$, where $N_s^{A_i,A_j}$ denote the number of sources they share.
\end{lemma}

The interpretation presented in \Cref{thm:VNE-interpretation} and \Cref{thm:MMI-interpretation} will be useful for proving the correctness of our protocol. Furthermore, in spirit the characteristic vector in~\cite{yang2022strong}, we define the \textit{characteristic matrix} of a quantum network to be
    \begin{equation}
        M_{\mathcal N} = \begin{pmatrix}
            S(A_1) & \mathcal I_m (A_1; A_2) & & \dots & & \mathcal I_m (A_1; A_{N_m}) \\
            \mathcal I_m (A_2; A_1) & S(A_2) & \mathcal I_m (A_2; A_3) &\dots& & \vdots \\
            \vdots & & &\dots& & \mathcal I_m (A_{N_m-1}; A_{N_m}) \\
            \mathcal I_m (A_{N_m}; A_1) & &  & \dots& \mathcal I_m (A_{N_m}; A_{N_m -1}) & S(A_{N_m})
        \end{pmatrix} \label{eq:char-mat}
    \end{equation}
where the diagonal is the characteristic vector $V_{\mathcal N}$ containing the von Neumann entropy and the off-diagonals are the measured mutual information.
Note that the matrix is symmetric, $M_{\mathcal N} = M_{\mathcal N}^\intercal$. By introducing the off-diagonal terms, we can quantify the number of sources two nodes share. This addition allows us to extend the previous network classification protocol to include cases where more than one entanglement is shared between nodes.

{We note that the characteristic matrix bears resemblance to the covariance matrix used in the semidefinite tests for network compatibility~\cite{Kela2020_semidefinite_tests,Aberg2020_semidefinite_tests, Kraft2021_characterizing_quantum_networks} where the off-diagonals of the covariance matrix are nonzero if and only if a source correlates two measurements.
A key distinction is that the covariance matrix is evaluated from the classical data sampled from the network, whereas the characteristic matrix is evaluated by optimizing the measurements with respect to von Neumann entropy and measured mutual information.
Indeed, it is significantly more efficient to evaluate the covariance matrix, however, having control over the measurement apparatus improves our ability to probe the strength of the correlation between measurement devices. Thus, the characteristic matrix might give a more detailed view of the network's topology.
We study the differences between entropic and covariance methods later in the manuscript (cf.~\Cref{sec:infer-network}) for inferring the topology of quantum networks using qubit measurements.}

We now show that the topology of an $n$-local quantum network can be fully characterized by the characteristic matrix $M_{\mathcal N}$, which specifies the von Neumann entropy at each node and the measured mutual information of each pair of nodes. 

\begin{theorem} \label{thm:noiseless-classifier}
Let two quantum networks, $\mathcal N \ind 1$ and $\mathcal N \ind 2$, satisfy \Cref{ass:all}. $\mathcal N \ind 1$ and $\mathcal N \ind 2$ have the same topology (cf. \Cref{def:isomorphism}) if and only if $S(A_i \ind 1) = S(A_i \ind 2)$ for all nodes $A_i$ and $\mathcal I_m(A_i \ind 1;A_j \ind 1) = \mathcal I_m(A_i \ind 2; A_j \ind 2)$ for all pairs of nodes $A_i, A_j$.
\end{theorem}

We defer the proof of the theorem to \Cref{sec:pf-noiseless-classifier}.

\begin{figure}
    \centering
    \includegraphics[width=.8\linewidth]{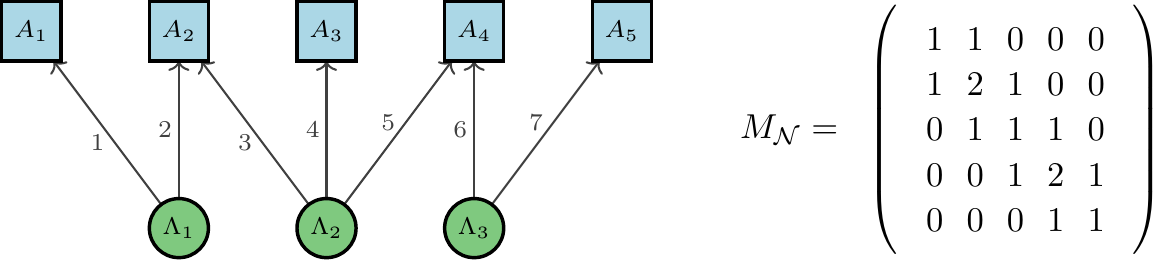}
    \caption{Example quantum network and its respective characteristic matrix $M_{\mathcal N}$. \Cref{thm:noiseless-classifier} shows that $M_{\mathcal N}$ uniquely characterizes a quantum network. However, inferring the network from $M_{\mathcal N}$ is nontrivial.}
    \label{fig:meas-infer-ex}
\end{figure}

The theorem states that one can uniquely find the topology of a network from its characteristic matrix. Suppose entropic quantities can be reliably calculated, verifying if two networks have the same topology requires only that the number of queries grows polynomially with respect to the number of nodes.

Inferring the topology from the characteristic matrix remains a difficult task. See \Cref{fig:meas-infer-ex} for an example network and its corresponding characteristic matrix $M_{\mathcal N}$. Knowing the topology, $M_{\mathcal N}$ can be straightforwardly obtained using \Cref{thm:VNE-interpretation} and \Cref{thm:MMI-interpretation}. However, we encourage the reader to try the other direction. Even though $M_{\mathcal N}$ indicates whether $A_i$ and $A_j$ share a source(s), finding the correct number of sources $N_s$ and assigning nodes to the respective sources appear to be highly nontrivial. Naively, one could search through all possible quantum networks with $N_s$ sources, for all possible $N_s$. The search space grows exponentially and would not be tractable for large networks. Thus, we defer the existence of a polynomial-time algorithm for decoding the characteristic matrix $M_{\mathcal N}$ as a future direction. 

\subsection{Distinguishing the topology of noisy networks}
Maximally entangled states such as GHZ states are fragile and easily corrupted by noise. In this section, we hope to establish robustness for the classification protocol introduced subjected to depolarizing noise, that is, for a quantum state $\rho$, a depolarizing channel $\mathcal E_\gamma$ performs the map
\begin{align}
    \mathcal E_\gamma (\rho) = (1-\gamma)\rho + \frac{\gamma}{2^n} \mathbb I_{2^n}
\end{align}
where $\gamma \in [0,1]$ is the parameter of the channel, and $n$ is the number of qubits involved in state $\rho$. In the quantum network setting, depolarizing noise acts jointly on qubits prepared by a source, and sources are affected independently.

\paragraph{Example: triangle network revisited}
Consider the triangle network 1 shown in \Cref{fig:tri-net-1}. Like in the noiseless case, the Shannon entropy at each measurement device is independent of the choice of measurement basis, and $H(A_i) = S(A_i) = 1$ for all measurement nodes in the network. In an attempt to characterize the topology, we look at the measured mutual information between devices $A_i$ and $A_j$, where the joint state is
\begin{align}
    \rho_{A_i \cup A_j} = \frac{\mathbb I_2}{2} \tensor \ket{\Phi} \bra \Phi \tensor \frac{\mathbb I_2}{2}.
\end{align}
If each qubit is sent through a depolarizing channel of the same noise parameter, then the joint state received $\mathcal E_\gamma (\rho_{A_i \cup A_j})$ becomes 
\begin{align}
    \mathcal E_\gamma (\rho_{A_i \cup A_j}) = \frac{\mathbb I_2}{2} \tensor \left( (1-\gamma) \ket \Phi \bra \Phi + \frac{\gamma}{4} \mathbb I_4 \right) \tensor \frac{\mathbb I_2}{2},
\end{align}
which yields the Shannon entropy of 
\begin{align}\label{eq:noisy-type-one}
    -H(\mathcal E_\gamma (\rho_{A_i \cup A_j})) = \frac{2-\gamma}{2} \log \left(2-\gamma\right) + \frac{\gamma}{2} \log \gamma - 4
\end{align}
when measured using the computational basis for both devices. When the channel is noiseless, i.e.~$\gamma = 0$, then $\mathcal I_m(A_i;A_j) = 1$ and we recover the results shown in \Cref{sec:topo-classification}. However, if the channel is completely noisy, meaning that $\gamma = 1$ , then $H(\mathcal E_\gamma(\rho_{A_i \cup A_j})) = 4$ and the measured mutual information is zero, which can lead one at the receiver end to think that qubits received at $A_i$ and $A_j$ are independent of one another. 

The same calculation can be applied to network 2 in \Cref{fig:network-example}. Let $\sigma_2 = (\ket{00}\bra{00} + \ket{11}\bra{11})/2$. Knowing that the joint system $\rho_{A_i \cup A_j} = \sigma_2 \tensor \sigma_2$, we focus on the behavior of one $\sigma_2$ knowing that the remaining system behaves identically and independently. We can find the noisy joint system of $\sigma_2$ to be
\begin{align}
    \mathcal E_\gamma (\sigma_2) = \frac{\gamma}{4} \mathbb I_4 + (1-\gamma)\sigma_2.
\end{align}
The above state yields the following Shannon entropy when measured in the computational basis.
\begin{align} 
    -H(\mathcal E_\gamma (\sigma_2)) = \frac{2-\gamma}{2} \log \left(2-\gamma\right) + \frac{\gamma}{2} \log \gamma - 2
\end{align}
Since the joint system is separable, the entropy of the joint system simply adds.
\begin{align}\label{eq:noisy-type-two}
    -H(\mathcal E_\gamma (\rho_{A_i \cup A_j})) = (2-\gamma) \log \left(2-\gamma\right) + \gamma \log \gamma - 4
\end{align}
Again, we can cover the noiseless and completely random behavior when $\gamma = 0$ or $\gamma = 1$, respectively. However, we are interested in determining if the two systems are ever indistinguishable due to noise. 

Claiming the computational basis is again the basis of choice that maximizes the mutual information and from \Cref{eq:noisy-type-one,eq:noisy-type-two}, the measured mutual information is, respectively,
\begin{align}
    \mathcal I_m  \ind 1 (A_i;A_j) &= \frac{2-\gamma}{2} \log \left(2-\gamma\right) + \frac{\gamma}{2} \log \gamma, \\
    \mathcal I_m \ind 2 (A_i;A_j) &= (2-\gamma) \log \left(2-\gamma\right) + \gamma \log \gamma.
\end{align}
We can see that the measured mutual information of the first network is half of the second one. Consequently, this means that for depolarizing of strength $\gamma \in [0, 1)$, the measured mutual information will always have a non-zero gap and we will be able to distinguish the two networks through local measurements. 

We extend \Cref{thm:noiseless-classifier} to the case of applying uniform global noise on sources. To do so, we first want to show that the same strategy of measuring in the computational basis in the noiseless case is still valid in the presence of noise (which is an extension of \Cref{thm:ent-lb,thm:ent-lb-2} in the appendix to the noisy case).

\begin{lemma}
Consider a quantum network satisfying \Cref{ass:all}, measurement nodes $A_i$ and $A_j$. Moreover, the network is made up of depolarizing channels that act on each qubit with strength $\gamma$. Then, the local measurement basis that maximizes the Shannon mutual information, that is
\begin{align}
    \underset{\{\Pi^{A_i}_{\vec a_i}\}, \{\Pi^{A_j}_{\vec a_j}\}}{\textnormal{argmax}} H(\mathbb P (\vec a_i)) + H(\mathbb P (\vec a_j)) - H(\mathbb P (\vec a_i, \vec a_j))
\end{align}
is the computational basis.
\end{lemma}

\begin{proof}
Let $\rho$ be either a Bell pair $\ket \Phi \bra \Phi$ or shared random bits $\sigma_2$. Under depolarizing noise on the source, the state becomes 
\begin{align}
    \mathcal E_\gamma(\rho) = (1-\gamma) \rho + \frac{1}{4} \mathbb I_4.
\end{align}
Note that for any unitary $U$ applied onto the noisy state, the effect of the noise stays unchanged, that is 
\begin{align}
    U \mathcal E_\gamma(\rho) U^\dagger = (1-\gamma) U \rho U^\dagger + \frac{1}{4} \mathbb I_4.
\end{align}
Thus, the measurement basis maximizing the Shannon entropy of $\rho$ remains the same for all $\gamma > 0$.
\end{proof}

We are interested in whether there exists a $\gamma \in [0,1)$ such that the measured mutual information of two networks is the same. If such $\gamma$ exists, denote $\gamma^*$, then the two networks are identical according to the protocol within some small neighborhood of $\gamma^*$. On the other hand, if the measured mutual information remains distinct for all $\gamma$ and pairs across two networks, then the protocol will be applicable for any noise level that is not completely depolarizing while assuming an infinite precision. Below, we provide one condition that guarantees such robustness to noise.

\begin{theorem} \label{thm:noisy-classifier}
Consider two quantum networks $\mathcal N \ind 1$ and $\mathcal N \ind 2$ satisfying \Cref{ass:all}. For depolarizing channels with known strength $\gamma \in [0,1)$, we can distinguish the topology of $\mathcal N \ind 1$ from $\mathcal N \ind 2$.
\end{theorem}
\begin{proof}
Under the assumption that each preparation node can send at most one qubit, for any two measurement nodes $A_i$ and $A_j$, the joint state $\rho_{A_iA_j}$ will be the tensor product of $\mathbb I_2$, $\sigma_2$, and Bell pairs. In particular, the von Neumann entropy of each device is invariant with respect to noise. Moreover, the measured mutual information of $\sigma_2$ and $\ket \Phi \bra \Phi$ is the same for all $\gamma \in [0,1]$. Let $I(\gamma) = \mathcal I_m (\mathcal E_\gamma (\ket \Phi \bra \Phi)) = \mathcal I (\mathcal E_\gamma (\sigma_2))$, the measured mutual information as a function of $\gamma$ will be $\mathcal I_m (A_i;A_j) = N_s^{A_i,A_j} I(\gamma)$. 

If $\mathcal N \ind 1$ and $\mathcal N \ind 2$ have different topologies, it means that there exists at least one pair of nodes $A_i$ and $A_j$ such that (without loss of generality) $N_s^{A_i \ind 1, A_j \ind 1} > N_s^{A_i \ind 2, A_j \ind 2}$. Thus, for any depolarizing channel of strength $\gamma$, $\mathcal I_m (A_i \ind 1; A_j \ind 1) > \mathcal I_m (A_i \ind 2; A_j \ind 2)$. Thus, by \Cref{thm:noiseless-classifier}, the two networks can be distinguished given sufficient shots taken to estimate each entropic quantity. 
\end{proof}

Note that \Cref{thm:noisy-classifier} actually provides a limited sense of noise robustness, that is, we need a priori knowledge of $\gamma$. The theorem can remain useful for comparing two quantum networks with unknown topologies but are under the influence of depolarizing noise of the same strength or verifying the topology of one noisy network. Recall that there are no straightforward algorithms for deriving the topology from the characteristic matrix even in the noiseless case. Moreover, inferring topology only becomes more difficult when the two networks are exposed to noises of different strengths. Furthermore, the assumption of depolarizing noise is largely due to simplicity, and we leave formal analysis of the protocol using other noise models as future work.

\section{Variational Quantum Network Topology Inference} \label{sec:variational-inference}

In this section we introduce our variational quantum optimization framework for inferring the topology of sources in noisy uncharacterized quantum networks. We implement our approach as Python software called \texttt{qNetTI}: The Quantum Network Topology Inferrer, which is publicly available on GitHub \cite{qnetti_2023}. Our software applies the \texttt{PennyLane} framework for cross-platform differential programming \cite{pennylane2018}, and builds upon \texttt{qNetVO}: the Quantum Network Variational Optimizer software \cite{qNetVO,Doolittle2023}. Our variational optimization techniques are hardware agnostic and can easily be extended to many general quantum network characterization, validation, and verification tasks. We apply our variational quantum optimization methods in Section~\ref{sec:infer-network}, in which we compare entropy-based and covariance-based topology inference when local qubit measurements are applied.

In general, when inferring the topology of sources in a quantum network, the respective covariances or entropic quantities must be sufficiently large such that correlations between measurements, or lack thereof, can be witnessed. However, since the measurement nodes do not have a standard reference frame for each of their measurement bases, a naive measurement choice will not necessarily reflect the correlations between measurement devices with accuracy. Thus, by optimizing the covariances or entropic quantities used to infer network topology, the network inference protocol is made more robust to the errors that result from nonoptimal measurement choices.

To further motivate our hardware-compatible optimization scheme, we note that the entropic quantities such as von Neumann entropy or measured mutual information contain an optimization implicitly in their operational definition (see \Cref{eq:vn_entropy_as_min_shannon_entropy,eq:meas-mut-info} respectively). Indeed, the von Neumann entropy at a given measurement device cannot be ascertained by performing a single measurement. That is, either the reduced density matrix measured at the device must be known, or the measurements must be optimized such that the Shannon entropy is minimized as in \Cref{eq:vn_entropy_as_min_shannon_entropy}. A similar remark can also be made for optimization used to obtain the measured mutual information in \Cref{eq:meas-mut-info}. To obtain these quantities in practice a quantum-hardware-compatible optimization technique is needed, hence we apply variational quantum optimization methods.

Variational quantum optimization is a type of hybrid quantum-classical algorithm in which a classical computer tunes a parameterized quantum circuit such that a cost function is minimized~\cite{Cerezo2021}. The quantum circuit is evaluated on quantum hardware leading to the hardware being optimized or trained for the task encoded by the cost function. Generally, variational quantum algorithms can be applied to a wide range of optimization and simulation problems and show promise of providing practical advantages~\cite{preskill2018quantum}. While their seminal applications were mainly in quantum computing, variational quantum algorithms have recently been proposed as a technique for optimizing noisy and uncharacterized quantum network hardware for various tasks \cite{Doolittle2023}. 

Within our hybrid optimization framework for quantum network topology inference, we construct a variational ansatz as follows. Let the collection of sources prepare the state $\ket{\psi} = \bigotimes_{i=1}^{N_s}\ket{\psi^{\Lambda_i}}$ while each measurement node performs the PVM measurement $\{\Pi^{A_j}_{a_j}(\thetav^{A_j})\}_{a_j\in\mc{A}_j}$ where
\begin{align}
    \Pi^{A_j}_{a_j}(\thetav^{A_j}) = U^{A_j}(\thetav^{A_j})^\dagger\op{a_j}{a_j}U^{A_j}(\thetav^{A_j})
\end{align}
and $U^{A_j}(\thetav^{A_j})$ is a unitary operator parameterized by $\thetav^{A_j} \in \mathbb{R}^m$. In total the measurement $\Pi_{\vec{a}}(\Theta) = \bigotimes_{j=1}^{N_m} \Pi^{A_j}_{a_j}(\vec{\theta}^{A_j})$ is applied while the probability of measuring outcome $\vec{a} = (a_j)_{j=1}^{N_m}$ is
\begin{align}
    \mathbb{P}(\vec{a}|\Theta) = |\ip{\psi|\Pi_{\vec{a}}(\Theta)|\psi}|^2.
\end{align}
Then, for measurement node $A_j$, we may rewrite the von Neumann Entropy in \Cref{eq:vn_entropy_as_min_shannon_entropy} as the cost function
\begin{equation}\label{eq:vn_entropy_cost}
    S(\rho_{A_j}) = \min_{\vec{\theta}^{A_j}\in \mathbb{R}^m} H(\mathbb{P}(a_j|\vec{\theta}^{A_j})).
\end{equation}
Likewise, measured mutual information  in \Cref{eq:meas-mut-info} becomes
    \begin{align}\label{eq:measured_mutual_info_cost}
        -\mathcal{I}_m(A_i ; A_j) = \min_{\vec{\theta}^{A_i} \in \mathbb{R}^{m_i}, \; \vec{\theta}^{A_j} \in \mathbb{R}^{m_j}} H(\mathbb{P}(a_i, a_j|\vec{\theta}^{A_i}, \vec{\theta}^{A_j})) - H(\mathbb{P}(a_i|\vec{\theta}^{A_i})) - H(\mathbb{P}(a_j|\vec{\theta}^{A_j}),
    \end{align}
in which the minus sign on the measured mutual information results from the convention of minimizing the cost function in variational optimization.

To optimize the quantities in \Cref{eq:vn_entropy_cost,eq:measured_mutual_info_cost}, a gradient descent algorithm is used. Consider a generic cost function $f(\Theta)$ that we aim to minimize. Then, we can incrementally approach the optimal settings $\Theta^\star = \arg \min_{\Theta} f(\Theta)$ by updating the settings as
\begin{equation}
    \Theta' = \Theta - \eta \nabla_{\Theta} f(\Theta),
\end{equation}
where $\eta > 0$ is a small stepsize and the gradient of $f(\Theta)$, $\nabla_{\Theta}f(\Theta)$, is a vector pointing in the direction of steepest ascent. Therefore in many small steps the algorithm navigates its way to a local minimum in the landscape of the cost function.

To apply these variational optimization methods to infer the topology of an unknown quantum state $\ket{\psi}$, it must be run on the quantum network's hardware. The reason is that $\ket{\psi}$ is not known and therefore cannot be reconstructed or characterized. To evaluate gradients on quantum hardware we make use of the parameter-shift rule \cite{Schuld2019_parameter_shift}, which evaluates gradients on quantum hardware with a computational complexity that is linear in the number of parameters. Hence, our variational methods for quantum network topology inference can conceivably be applied with efficiency in practical quantum networking systems.

As a final remark, we note that our variational quantum network topology inference scheme provides several improvements to existing techniques. First, previous methods do not typically consider the measurements to be tunable. By optimizing measurements, our variational methods could improve the effectiveness of existing topology inference approaches in which the violation of entropic bounds or Bell inequalities are used to test for network topology \cite{Henson2014_entropic_bounds,chaves2015_entroptic_bounds,Weilenmann2017_entropic_bounds,Chaves2016_polynomial,Tavakoli2016tree,Rosset2016_nonlinear_bell_inequalities,Tavakoli2022_network_nonlocality,Renou2019_correlation_limits,Luo2021network_configuration,Wolfe2019_inflation,Wolfe2021_inflation}. Second, quantities such as the von Neumann entropy or measured mutual information cannot be obtained for unknown quantum states. Typically, calculating these quantities would require estimating the entire quantum state via state tomography. Alternatively, our quantum-hardware-compatible variational methods can approximate these quantities through optimization. Thus, entropy-based inference methods can be applied in practice. Finally, our methods are hardware agnostic. The parameter-shift rule requires only a parameterized description of the unitary applied by the measurement device, how the unitary is physically implemented is not relevant. 

\section{Inferring network topology from qubit measurements} \label{sec:infer-network}

Previously, we considered the qubits in a measurement node to be measured jointly. While there exists a unique mapping from network topology to the characteristic matrix (cf.~\Cref{thm:noiseless-classifier}), constructing the map in a reasonable amount of time, particularly for large networks, is difficult. Moreover, qubits undergoing depolarizing channels with unknown noise parameters can render the protocol useless. Thus, the classification protocol introduced in \Cref{sec:distinguish-network} would be of practical use only in very limited situations.

These issues can be resolved if measurements are available on the qubit level, which is often the case for real-life situations. The increased granularity allows for a protocol that can infer network topology from measurements local to each node in polynomial time. Moreover, the procedure no longer requires GHZ state preparations, is entirely robust to noise induced by quantum channels, and can maintain performance using low numbers of circuit evaluations.

The protocol is also not restricted to measured mutual information to estimate classical correlation. In particular, we show that covariance and classical mutual information are both suitable choices for identifying the entanglement structure of quantum networks. While the covariance-based methods generally exhibit more stable convergence and are more robust to noise, entropic alternatives excel under the low shot regime.

\subsection{Entropy-based protocol for inferring network topology}

We can derive a protocol for inferring network topology by straightforwardly extending the characteristic matrix $M_{\mathcal N}$ (\Cref{eq:char-mat}) to qubits. Let $\{q_i$, $i \in [N_q]\}$, be the set of qubits in the network. Then, we can define the \textit{qubit characteristic matrix}:
    \begin{equation}
        Q_{\mathcal N} = \begin{pmatrix}
            S(q_1) & \mathcal I_m (q_1; q_2) & & \dots & & \mathcal I_m (q_1; q_{N_q}) \\
            \mathcal I_m (q_2; q_1) & S(q_2) & \mathcal I_m (q_2; q_3) &\dots& & \vdots \\
            \vdots & & &\dots& & \mathcal I_m (q_{N_q-1}; q_{N_q}) \\
            \mathcal I_m (q_{N_q}; q_1) & &  & \dots& \mathcal I_m (q_{N_q}; q_{N_q -1}) & S(q_{N_q})
        \end{pmatrix}. \label{eq:qubit-char-mat}
    \end{equation}
Similar to $M_{\mathcal N}$, the diagonal entries stores the von Neumann entropy of each single-qubit state and the off-diagonal entries stores the correlation quantified by the measured mutual information. By treating each qubit as its own measurement node, the qubit characteristic matrix inherits all the properties of $M_{\mathcal N}$, and gains additional structure that can help with decoding the entanglement structure.

\begin{theorem}\label{thm:qubit_network_characterization}
    Consider an $n$-local network $\mathcal N$ measured using local qubit projectors $\Pi^{\mathcal N}_{\vec{a}} = \bigotimes_{j=1}^m \Pi_{a_j}^{q_j}$ where $a_j\in\{0,1\}$ and $q_j\in[N_q]$ index the measured qubit. 
    Moreover, suppose \Cref{ass:all} holds. Then, the network's topology is completely characterized by the qubit characteristic matrix $Q_{\mathcal N}$ (\Cref{eq:qubit-char-mat}), where the $i^{th}$ row lists the qubits entangled with qubit $q_i$ and the number of sources $N_s$ is equivalent to the number of unique rows (or columns) of $Q_{\mathcal N}$.
\end{theorem}

\begin{proof}
Consider each qubit to be its own measurement node. Then, the matrix $Q_{\mathcal N}$ completely characterizes the network by invoking \Cref{thm:noiseless-classifier}.

To see the relationship between the number of unique rows (or columns) by the symmetry of $Q_{\mathcal N}$, fix a particular source $\Lambda_k$. Then, for any row $i$ where $q_i \in \Lambda_k$, the $(i,j)$-th entry is $1$ if and only if $q_j \in \Lambda_k$. And since each qubit can only come from one source, all qubits whose corresponding rows are equal to one another are from the same source.  Therefore, partitioning the set of qubits into sets whose respective rows are equal recovers the qubit-source mapping, and the number of partitions is the number of sources.
\end{proof}

It is important to note that we don't need \Cref{subass:A1} of \Cref{ass:all} to hold to infer the topology. If we construct another matrix $B$ defined component-wise by 
\begin{align}
    B_{i,j} = \begin{cases}
    1 &\IF Q_{\mathcal N, i, j} > 0 \\
    0 &\OW
    \end{cases}
\end{align}
and treat each qubit as individual measurement nodes, then $B$ is equivalent to $Q_{\mathcal N}$ with the same entanglement structure assuming GHZ state preparation. To infer the topology, we can simply apply \Cref{thm:qubit_network_characterization} to the matrix $B$. Thus, for the remainder of the section, we use $Q_{\mathcal N}$ to refer to both the qubit characteristic matrix in \Cref{eq:qubit-char-mat} while assuming GHZ preparation and this binary matrix $B$ without the state preparation assumption as they are functionality equivalent.

\begin{figure}
    \centering
    \includegraphics[width=.8\linewidth]{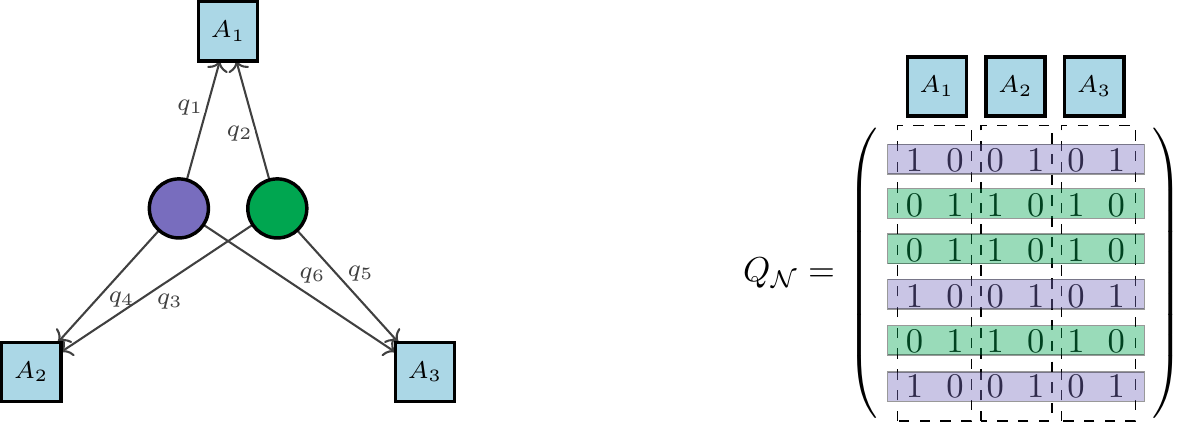}
    \caption{Application of \Cref{thm:qubit_network_characterization} on the network in \Cref{fig:tri-net-2}. The columns of $Q_{\mathcal N}$ are organized by the nodes each qubit is from, and unique rows are grouped together. Each group of unique rows corresponds to a preparation node (purple or green) and the connectivity can be found by observing the non-zero entries in each row.}
    \label{fig:qubit-infer-ex}
    
\end{figure}

The above theorem gives an algorithm for reconstructing the network topology given the qubit-wise characteristic matrix $Q_{\mathcal N}$. First, assume knowledge of the measurement node that each qubit is sent to. The columns of $Q_{\mathcal N}$, each representing a qubit, can then be grouped into their respective measurement nodes. On the other hand, the rows of $Q_{\mathcal N}$ can be partitioned into sets of indices with identical rows, that is, let $\Lambda_i$ be a set such that for all $j,k \in \Lambda_i$, $Q_{\mathcal N, j,*} = Q_{\mathcal N, k,*}$. As the notation suggests, this set of indices is the set of qubits from the source $\Lambda_i$. Lastly, if $Q_{\mathcal N,r,s} = 1$, qubits $r$ (in node $A_j$) and $s$ (in node $A_k$) share a source $\Lambda_i$ and the triplet $(A_j,\Lambda_i,A_k)$ exists in the network. In \Cref{fig:qubit-infer-ex}, we give an elementary demonstration of the algorithm on the triangle network presented in \Cref{fig:tri-net-2}.

Performing joint measurements, even for qubits received in the same measurement node, can be experimentally demanding. Thus, limiting ourselves to qubit-wise measurement actually enhances the practicality of the protocol. The fine-grained measurements provide a simple algorithm for determining the topology in time quadratic to the number of qubits. Moreover, sources are not restricted to preparing only GHZ states; any entangled states will do. Given infinite precision and noiseless channels, each entry of $Q_{\mathcal N}$ can store whether there exist correlations between qubits. As correlations can only arise from quantum entanglement, two qubits are from the same source if and only if a non-zero correlation is observed. 

Since the protocol aims at identifying zero or non-zero correlations, as long as the quantum channel is not completely destructive, one can always infer the topology given enough shots to suppress statistical noise. This statement is formalized below.

\begin{theorem}\label{thm:qubit-network-characterization-noisy}
    Consider a noisy network $\mathcal N$ that is measured using local qubit measurements and $\rhoNet = \bigotimes_{i=1}^n \rho^{\Lambda_i}$.
    Its topology is completely characterized by the matrix $Q_{\mathcal N}$ as described in \Cref{thm:qubit_network_characterization}.
    \begin{proof}
        Let $\rho_i = \tr_{j\neq i}[\rhoNet]$ for all $j\in[N_q]$. Then, if $S(\rho_i) > 0$ a source may exist that correlates the qubit with other qubits.
        Next, the measured mutual information $\mathcal I_m(q_i,q_j) = 0$, if and only if $\rho_i\otimes\rho_j$. This implies that when $I_m(q_i,q_j) > 0$, a source must be present to correlate the two independent qubit measurements.
        Therefore, the matrix $Q_{\mathcal N}$ only has nonzero elements on its off-diagonal if there exist sources to correlate the qubits.
        In practice, finite samples are taken and the scalar value of $\mathcal I_m(q_i,q_j)$ should only be counted as nonzero if it is sufficiently larger than the statistical fluctuations of uncorrelated qubits.
    \end{proof}
\end{theorem}

Note that in the above Theorem, a source can be so noisy that it separates as $\bigotimes_j \rho_j$.
We argue that the source no longer qualifies as such precisely because it no longer distributes shared randomness.
Otherwise, so long as a sufficient number of measurements are taken and non-separable states are prepared at each source, the qubit-wise characteristic matrix $Q_{\mathcal N}$ is sufficient for determining the network topology.
On a different note, for certain choices of noise, such as colored noise and qubit dephasing noise, the characteristic matrix is preserved; thus, both \Cref{thm:noisy-classifier} and \Cref{thm:qubit-network-characterization-noisy} will hold under these noise models.

\paragraph{Alternative definition for mutual information on networks} 
Constructing the characteristic matrix $Q_{\mathcal N}$ requires the number of independent circuit evaluations to grow at a rate equal to the number of unique qubit-pairs. We associate a unique measurement basis that achieves the maximum of \Cref{eq:meas-mut-info} for each pair. In the case of $N_q$ being large, we might be tempted to define a measurement basis common across all qubit pairs that maximizes an analogous quantity that takes the mutual information of all pairs into account. Formalizing this intuition, we seek measurement operators $\{\Pi_{\vec x} = \bigotimes_{q_i} \Pi^{q_i}_{x_i}\}$ such that 
\begin{align}
    \max_{\{\Pi_{\vec x}\}} \sum_{i < j} H(\mathbb P(x_i)) + H(\mathbb P(x_j)) - H(\mathbb P(x_i, x_j)) \label{eq:classical-mi}
\end{align}
and we let the mutual information between any two qubits obtained from such basis be the \emph{classical mutual information} (as opposed to the original \emph{bipartite} measured mutual information).

Comparing \Cref{eq:classical-mi} with the expression for measured mutual information in \Cref{eq:meas-mut-info}, we can see that
\begin{align}
    &\max_{\{\Pi_{\vec x}\}} \sum_{i < j} H(\mathbb P(x_i)) + H(\mathbb P(x_j)) - H(\mathbb P(x_i, x_j)) \nonumber \\ 
    &\leq \sum_{i < j} \max_{\{\Pi_{x_i}^{q_i} \tensor \Pi_{x_j}^{q_j}\}} H(\mathbb P(x_i)) + H(\mathbb P(x_j)) - H(\mathbb P(x_i, x_j)).
\end{align}
Thus, for each unique pair of qubits, measured mutual information between the qubit pair is at least that of the one calculated from the classical mutual information of the whole network. We can construct an example where the inequality is, in fact, strict.

\begin{remark}\label{thm:network-vs-bipartite}
Consider a tripartite network that has the following mixed state:
\begin{align}
\rho = \frac{1}{8} \biggr( &\mathbb I \tensor \mathbb I \tensor \mathbb I + \frac{1}{2} \sigma_x \tensor \sigma_x \tensor \mathbb I \nonumber \\
&+ \frac{1}{2} \sigma_y \tensor \mathbb I \tensor \sigma_y + \frac{1}{2} \mathbb I \tensor \sigma_z \tensor \sigma_z \biggr).
\end{align}
For this network, the measurement basis that maximizes the correlation---mutual information and covariance (see next section) alike---between any pair of qubits is orthogonal to that of any other pair. Since the bipartite measured mutual information consists of a collection of measurement bases, one for each pair, correlations can be fully observed. In the case of classical mutual information, in which we seek to find one optimal basis for all qubit pairs, there will always be correlations that are not fully observed.
\end{remark}

On the other hand, the benefit of considering the network mutual information is clear: the number of circuit evaluations needed is constant with respect to the number of qubits.

\paragraph{Covariance-based topology inference}

The measured mutual information is not the only measure of the correlation between two random variables. Inspired by the literature~\cite{Kela2020_semidefinite_tests,Aberg2020_semidefinite_tests,Kraft2021_characterizing_quantum_networks}, we propose a network topology inference protocol by applying the decoding scheme above on covariance matrices.

We treat each qubit measurement as a random variable taking values in $\{-1, 1\}$. From the definition of covariance, we get
\begin{align}
    &\Cov(q_i, q_j) \nonumber \\ 
    &= \sum_{x, y \in \{-1,+1\}} x y ~\mathbb P(q_i = x, q_j = y) - \bar q_i \bar q_j \\
    &= \sum_{x, y \in \{-1, +1\}} x y ~ \tr \left( \left(\Pi_x^{q_i} \tensor \Pi_y^{q_j}\right)  \rho_{q_i q_j} \right) - \bar q_i \bar q_j
\end{align}
where $\rho_{q_i q_j}$ is the state of qubits $i$ and $j$ and 
\begin{align}
    \bar q_i &= \sum_{x \in \{-1, +1\}} x~\mathbb P(q_i = x) \\
    &= \sum_{x \in \{-1,+1\}} x~\tr \left( \Pi_x^{q_i}~\tr_j(\rho_{q_i}) \right),
\end{align}
and similarly for $\bar q_j$. The qubit covariance matrix, denoted by $C_{\mathcal N}$, is defined entry-wise where the $(i,j)$-th entry contains the absolute value of the covariance between qubits $i$ and $j$.

The magnitude of the covariance depends on the measurement basis. If the basis is chosen arbitrarily, there is a chance that the basis of choice yields low covariance, thereby skewing the inference. To reduce the probability of such an event from occurring, we can find the basis that maximizes the distance between the covariance matrix and the origin, for example, for some choice of local measurement basis $\{\Pi_{\vec x} = \bigotimes_{q_i} \Pi^{q_i}_{x_i}\}$, we hope to find
\begin{align}
    \arg \max_{\{ \Pi_{\vec x} \}} ~\tr \left( C_{\mathcal N}^\dagger C_{\mathcal N} \right).
\end{align}
However, note that \Cref{thm:network-vs-bipartite} still applies and there exist correlations that cannot be fully observed from the covariance matrix alone. However, this does not interfere with inferring the network topology.

From here, the procedure for inferring the network topology is identical to that of the mutual information-based protocol. We provide a binary description for each pair of qubits in a network---$1$ if the covariance is above zero, and $0$ (or within an acceptable statistical margin of error) otherwise. If two qubits are entangled, then their covariance is nonzero as long as the two qubits are not measured in mutually unbiased bases, for example, $\sigma_x$ and $\sigma_z$. Thus, the protocol straightforwardly applies by replacing the qubit characteristic matrix $Q_{\mathcal N}$ with the covariance matrix $C_{\mathcal N}$. 

\begin{remark} 
The 
covariance methods presented in \cite{Kela2020_semidefinite_tests,Aberg2020_semidefinite_tests,Kraft2021_characterizing_quantum_networks} aim at constructing a single covariance matrix after collecting classical measurement data. In comparison, our protocol sequentially queries the network and finds the optimal covariance matrix  that maximizes the correlation observed among all pairs of qubits. Our method avoids the possibility of choosing a measurement basis that leads to a small observable correlation, though at the cost of requiring more extensive procedures like variational optimization.
\end{remark}

\subsection{Comparing entropic and covariance-based protocols}

While the mutual information and covariance both quantify the amount of correlation between qubits, each method possesses unique attributes that make one more preferable than the other depending on the context. In the following section, we will compare the two methods in terms of their computational complexity, the ability to accurately infer the topology under noisy quantum channels, and under statistical noise from taking a finite number of measurements. We will also demonstrate the capability of both methods on quantum hardware.

\paragraph{Computational complexity}

The number of circuit evaluations needed to evaluate the covariance matrix is constant in the number of qubits. Constructing the characteristic matrix using classical mutual information needs twice the number of evaluations as the diagonal and off-diagonal entries must be constructed independently. Lastly, in the case of using the bipartite measured mutual information, the number of circuit evaluations grows quadratically in the number of qubits. One can optimize this to achieve linear scaling by evaluating all non-overlapping pairs, but is still time-consuming compared to the two constant-time alternatives.

The time complexity needed to compute the cost functions ranks similarly. Computing the distance induced by the Schmidt norm of the covariance matrix from the origin requires scanning through each entry of the matrix, resulting in a quadratic scaling with respect to the number of qubits. Computing mutual information requires specifying the marginal distribution of all qubit pairs. One can think of the measurement outcome as a probability distribution with a domain that grows exponentially with the number of qubits. Obtaining the marginal distribution for a pair of qubits requires summing over exponentially many terms, and therefore is a time-consuming procedure. In the case of measured mutual information, one can choose to query one pair of qubits at a time, sacrificing query complexity for ease of computing the cost function.

\paragraph{Inference under noisy quantum channels}\label{sec:noisy}

\begin{figure}
    \centering
    \includegraphics[width=.75\linewidth]{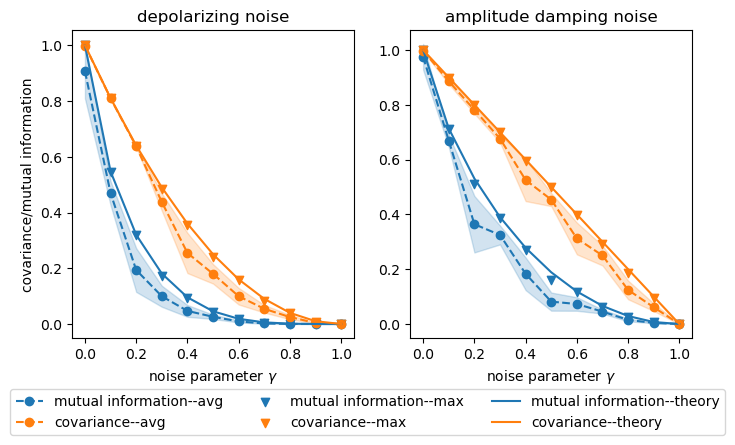}
    \caption{The effect of noisy quantum channels on the observed correlation. A Bell state with fixed local rotations was prepared, and we applied \textbf{(left)} depolarizing noise to each qubit and \textbf{(right)} amplitude-damping noise to each qubit using \texttt{PennyLane}'s mixed-state simulator. We ran 10 independent trials of variational optimization with step sizes of $0.05$ for $30$ steps, plotting the average (circles with shaded standard error) and maximum (triangle) observed across trials. We also compare empirical value obtained through simulation with theory (solid lines), cf.~\Crefrange{eq:depolarizing-mi}{eq:depolarizing-cov} and \Crefrange{eq:ampdamp-mi}{eq:ampdamp-cov}}
    \label{fig:channel-noise}
\end{figure}

Though conceptually identical, covariance is generally less affected by noisy quantum channels compared to mutual information. Consider a Bell state undergoing depolarizing channels applied qubit-wise. We can then write the mutual information and covariance analytically in terms of the noise parameter $\gamma \in [0, 1]$.
\begin{align}
    I_m(q_1;q_2) &= \frac{1 + (1-\gamma)^2}{2} \log(1 + (1-\gamma)^2) \nonumber \\
    &\quad + \frac{\gamma(2-\gamma)}{2} \log(\gamma (2-\gamma)), \label{eq:depolarizing-mi}\\
    \Cov(q_1,q_2) &= (1-\gamma)^2 \label{eq:depolarizing-cov}.
\end{align}
The same can be done for amplitude-damping noise applied onto each qubit, which has Kraus operators
\begin{align}
    K_0 = \begin{pmatrix} 1 & 0 \\ 0 & \sqrt{1-\gamma} \end{pmatrix}, ~~ K_1 = \begin{pmatrix} 0 & \sqrt{\gamma} \\ 0 & 0 \end{pmatrix}
\end{align}
for each qubit, and noise parameter $\gamma \in [0,1]$. The maximal correlation is observed when measuring both qubits in the $\sigma_x$ basis, which respectively takes the form
\begin{align}
    I_m(q_1;q_2) &= \frac{2-\gamma}{2} \log (2-\gamma) + \frac{\gamma}{2} \log \gamma, \label{eq:ampdamp-mi}\\
    \Cov(q_1, q_2) &= 1 - \gamma \label{eq:ampdamp-cov}.
\end{align}
More detailed calculations are deferred to \Cref{sec:noisy-channel-calc}.
We compare these theoretic quantities with numerical solutions obtained via variational optimization. We repeated 10 independent optimization trials per choice of noise parameter, and the results are shown in \Cref{fig:channel-noise}.

Both analytical and numerical results show that covariance decays slower with respect to the noise parameter $\gamma$ in the case of depolarizing and amplitude-damping noise. This suggests a noise regime where mutual-information-based protocols misidentify the existence of entanglement while  the covariance matrix can accurately infer the network topology. On a more optimistic note, we observe that the respective theoretical values can be achieved numerically given an ensemble of independent optimization trials, and a typical optimization run produces a result that approximates the theoretical.

\paragraph{Inference under finite shot noise}

In addition to noise induced by quantum channels, statistical shot noise is another practical source of uncertainty. In particular, having computational methods that can operate in the low-shot regime greatly increases its applicability as quantum resources are often scarce. Shot error can pose an issue for iterative methods such as variational optimization as noisy gradient estimates can limit the ability to effectively navigate the optimization landscape. We put our inference protocol to the test with a 5-qubit network where the first three qubits are entangled in a W state, i.e., the $W$ state takes the form
\begin{align}\label{eq:w_state}
    \ket{W} = \frac{1}{\sqrt{3}} \left( \ket{001} + \ket{010} + \ket{100} \right),
\end{align}
and the remaining two in a GHZ state. The optimization result for qubit pairs with and without entanglement is shown in \Cref{fig:finite-shot}.

\begin{figure*}
    \centering
    \includegraphics[width=\linewidth]{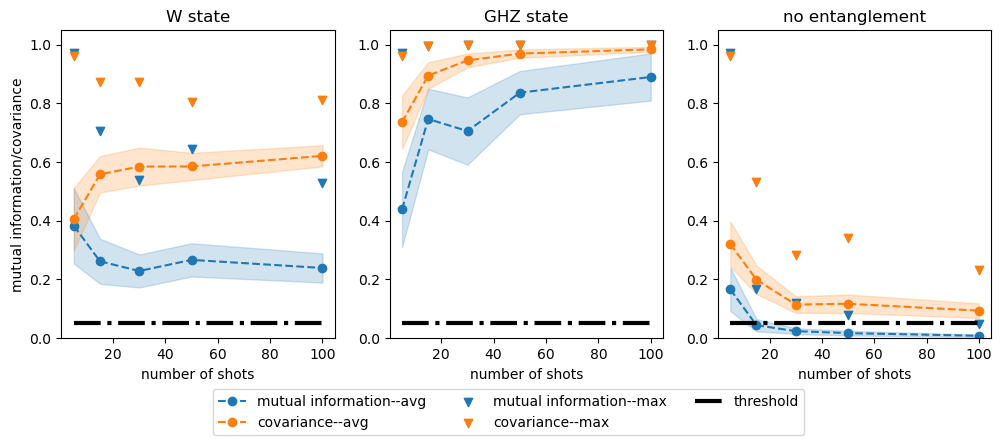}
    \caption{The effect of statistical shot noise on calculating correlations under various prepared states. Mutual information and covariance are estimated using variational optimization with a step size of $0.05$ for 30 steps. Each plot shows the correlation between qubit pairs under different entanglement structures. The threshold value (black line) is set to $0.05$, and was chosen arbitrarily. From the experiment, we observe that covariance methods accurately identify the existence of entanglement, as shown in the case of \textbf{(left)} W states (\Cref{eq:w_state}) and \textbf{(middle)} GHZ states (\Cref{eq:ghz}). However, entropic methods are able to identify \textbf{(right)} the lack of entanglement better than covariance-based protocols.}
    \label{fig:finite-shot}
\end{figure*}

Consistent with the results in \Cref{fig:channel-noise}, covariance methods identify the presence of entanglement more strongly than entropic measures. However, covariance tends to ``overshoot'' and falsely declares the presence of entanglement for uncorrelated qubits. We can set an arbitrary threshold for detecting correlations---we decide that two qubits are entangled if and only if the calculated correlation (mutual information or covariance) is above the defined threshold. In the case of figure \Cref{fig:finite-shot}, the threshold is set to 0.05, indicated by the black line. When qubits are not entangled, mutual information rapidly vanishes while covariance stays above the threshold within the range of the shot counts considered. Thus, in the low-shot regime, mutual information might be more suitable as it is more likely to correctly reconstruct the entanglement structure of the underlying network.

\paragraph{Hardware experiments}

\begin{figure}[t!]
    \centering
    \includegraphics[width=0.95\textwidth]{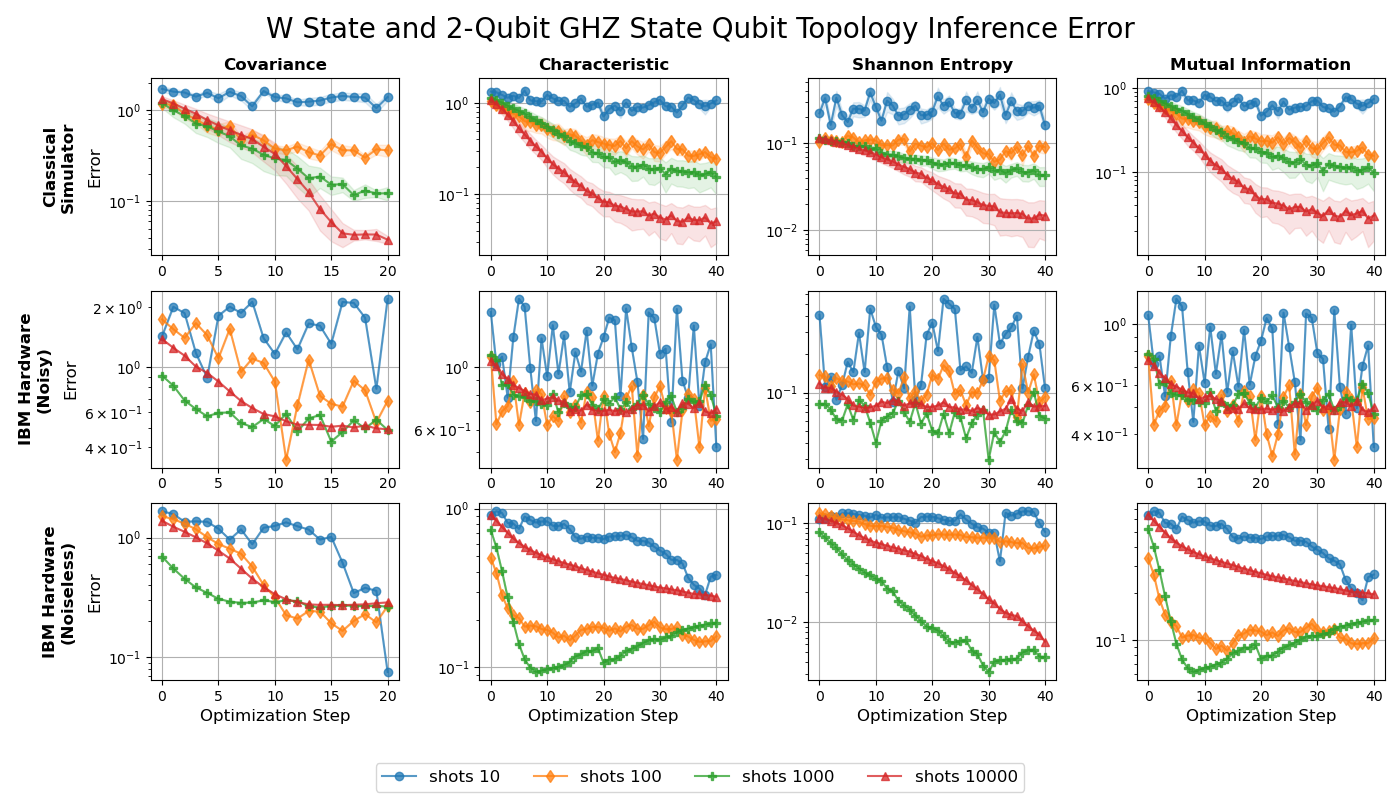}
    \caption{Variational quantum optimization of the covariance matrices and characteristic matrices for a W state and two-qubit GHZ state prepared on IBM Hardware. From left to right each column plots the covariance matrix optimization, the characteristic matrix optimization, the Shannon entropy minimization, and the classical mutual information maximization. Note that the Shannon entropy and mutual information optimizations are combined to construct the characteristic matrix. In each plot the blue circles show the 10-shot optimization, the orange diamonds show the 100-shot optimization, green plus signs show the 1,000-shot optimization, and red triangles show the 10,000-shot optimization. The $x$-axis shows the optimization step while the $y$-axis shows the inference error calculated as the Euclidean distance between the ideal covariance/characteristic matrix and the matrix in each optimization step. The first row shows the optimization data averaged over 10 runs on a finite-shot noiseless classical simulator, the second row shows the optimization data collected from the \texttt{ibmq\_belem} quantum computer, and the third row shows the data collected when the settings from the noisy IBM hardware optimization are reevaluated on a noiseless classical simulator.}
    \label{fig:w_state_ibm_hardware_inference_optimization}
\end{figure}

\begin{figure}[b!]
    \centering
    \includegraphics[width=0.95\textwidth]{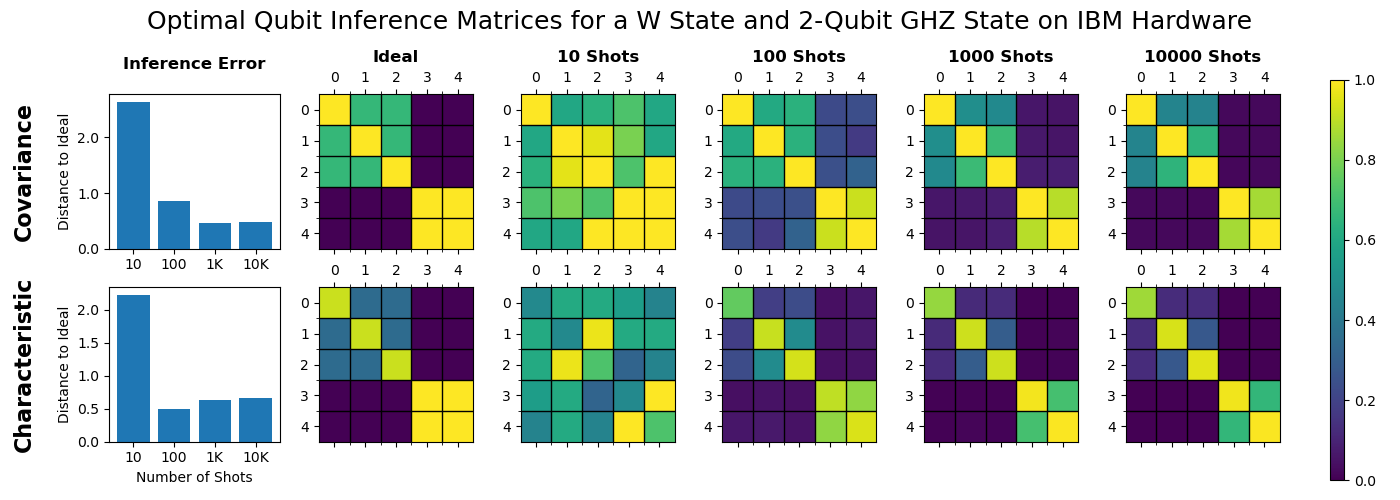}
    \caption{In the top row, we plot the maximal value achieved for each term of the covariance matrix across all optimization steps. In the bottom row, we plot the maximal mutual information (off-diagonals) and minimal Shannon entropies (diagonal) across all optimization steps. The bar graph plots the Euclidean distance between the ideal and inferred matrices for each distinct number of shots.}
    \label{fig:w_state_ibm_optimal_inference}
\end{figure}

We now apply on IBM's quantum hardware our variational scheme for inferring network topology using local qubit measurements. As an example we again consider the 5-qubit state preparation $\ket{\psi} = \ket{W}\otimes \ket{\Phi}$ where the $\ket{W}$ is the W state in \Cref{eq:w_state} and $\ket{\Phi} = \frac{1}{\sqrt{2}}(\ket{00} + \ket{11})$ is a two-qubit maximally entangled state. We also apply our variational network inference scheme to 5-qubit GHZ states and the 5-qubit zero state (see \Cref{appendix:ibm_inference_data_5-qubit_zero_and_ghz}). In each example, a known state is prepared while we optimize a variational ansatz that parameterizes arbitrary qubit projective measurements as $\Pi_x = U^\dagger(\vec{\theta})\op{x}{x} U(\vec{\theta})$ where $\vec{\theta}\in \mathbb{R}^3$. We then compare the performance of our variational inference scheme on both noisy IBM hardware and noiseless classical simulator when 10, 100, 1,000, and 10,000 shots are considered.

To quantify the performance of the optimization, we consider the inference error, which we define as the Euclidean distance
\begin{equation}\label{eq:euclidean_distance_inference_error}
    d(C,C^\star) = \sqrt{\tr\left[ \left( C^\star - C \right)^\intercal \left( C^\star - C \right) \right]}
\end{equation}
where $C$ is the optimized covariance matrix and $C^\star$ is the ideal covariance matrix for the given state. The distance in \Cref{eq:euclidean_distance_inference_error} can similarly  quantify the error in the characteristic matrix $d(Q,Q^\star)$. Note that this inference error quantifier only works in our numerical experiment because we prepare a known state. In practice, the state preparation is not known and the performance of the optimization cannot be quantified by \Cref{eq:euclidean_distance_inference_error}.

For the considered state $\ket{\psi} = \ket{W}\otimes\ket{\Phi}$, the ideal qubit covariance and characteristic matrices are
\begin{equation}\label{eq:ideal_w_state_ghz_state_inference_matrices}
    C^\star = \begin{pmatrix}
        1 & \frac{2}{3} & \frac{2}{3} & 0 & 0 \\
        \frac{2}{3} & 1 & \frac{2}{3} & 0 & 0 \\
        \frac{2}{3} & \frac{2}{3} & 1 & 0 & 0 \\
        0 & 0 & 0 & 1 & 1 \\
        0 & 0 & 0 & 1 & 1 \\
    \end{pmatrix}\!, Q^\star = \begin{pmatrix}
        S_W & I_W & I_W & 0 & 0 \\
        I_W & S_W & I_W & 0 & 0 \\
        I_W & I_W & S_W & 0 & 0 \\
        0 & 0 & 0 & 1 & 1 \\
        0 & 0 & 0 & 1 & 1 \\
    \end{pmatrix}
\end{equation}
where $I_W \equiv \mathcal{I}_m(q_i;q_j) \approx 0.349976$ and  $S_W \equiv S(W^{q_i}) \approx 0.918296$ for any of the qubits $q_i$ and $q_j$ of the state $\ket{W}$. To obtain the ideal covariance matrix and mutual information in \Cref{eq:ideal_w_state_ghz_state_inference_matrices}, it is sufficient to measure all qubits in the $\{\ket{+},\ket{-}\}$ basis. While this calculation is straightforward for the $\ket{\Phi}$ state, we will be more explicit with the $W^{ABC} = \op{W}{W}$ state, for which the reduced density matrices are
\begin{align}\label{eq:vn_entropy_w_state_qubit}
    W^A = \tr_{BC} \left( W^{ABC} \right) = \frac{2}{3}\op{0}{0} + \frac{1}{3}\op{1}{1}
\end{align}
and 
\begin{align}
    W^{AB} = \tr_C \left( W^{ABC} \right) = \frac{1}{3} \ket{00}{00} + \frac{2}{3} \op{\Psi^+}{\Psi^+}
\end{align}
where $\ket{\Psi^+} = (\ket{01} +\ket{10})/\sqrt{2}$. When the observable $\sigma_x$ is measured, the covariance and variance are
\begin{align}
    \text{Cov}(A,B) = \tr \left( (\sigma_x \otimes \sigma_x)~W^{AB} \right) = \frac{2}{3}   
\end{align}
and $\text{Var}(A) = 1$. The qubit von Neumann entropy can then be explicitly calculated from \Cref{eq:vn_entropy_w_state_qubit} as 
\begin{align}
    S_W \equiv S(W^A) = -\frac{2}{3}\log \frac{2}{3}  - \frac{1}{3} \log \frac{1}{3} \approx 0.918, 
\end{align}
while the measured mutual information is found to be $I_W = \mathcal{I}_m(A,B) \approx 0.349976$ where each party measures the observable $\sigma_x$. 
 
Since all qubits are measured in the same basis, the measured mutual information is equivalent to the classical mutual information for all qubit pairs, that is $\mathcal{I}_m(q_i,q_j) = I(q_i,q_j)$. Thus, for this example, it is sufficient to optimize the classical bipartite mutual information rather than the measured mutual information. This simplification provides significant speed-ups when running on quantum hardware because only one circuit needs to be evaluated to collect the mutual information of all qubit pairs as opposed to the 10 circuits needed to optimize the measured mutual information for each qubit pair independently. Note that queue times for the IBM hardware are the main bottleneck in our variational network inference scheme.

In \Cref{fig:w_state_ibm_hardware_inference_optimization}, we plot the inference error as the covariance and characteristic matrices are optimized with respect to the state $\ket{\psi} = \ket{W}\otimes \ket{\Phi}$. To investigate the relation between the number of shots and the inference error, we consider optimizations where 10, 100, 1,000, and 10,000 shots are used when collecting data from the quantum computer. We expect that the inference error should decrease as the number of shots increases.

As a baseline, we first run our numerical experiment on \texttt{PennyLane}'s \texttt{default.qubit} classical simulator, which is a noiseless, finite-shot simulation of a quantum computer. The data shown in the top row of \Cref{fig:w_state_ibm_hardware_inference_optimization} plots the mean inference error from 10 independent optimizations for both the covariance and characteristic matrices. As expected, the amount of inference error decreases as the number of shots increases.  Furthermore, we find that the covariance and characteristic matrices do not always find a global optimum in which all matrix terms converge to their maximal theoretical values. As a result, the mean does not approach zero, reflecting the importance of running the optimization algorithm multiple times. Moreover, we find that the covariance matrix is optimized in fewer iterations than the characteristic matrix.

We then test our optimization on the \texttt{ibmq\_belem} 5-qubit quantum computer, which exhibits a considerable amount of noise.  When we run the optimization on an IBM quantum computer, the optimization steps become expensive to run due to the queue wait times. As a result, we are only able to run one optimization for each number of shots. We plot the IBM hardware optimization results in the middle row of \Cref{fig:w_state_ibm_hardware_inference_optimization}. For the 10-shot case, it is not clear that the inference error is decreasing. For larger numbers of shots, the inference error decreases throughout the optimization. To further validate the optimization results, we take the settings optimized on the noisy IBM hardware and re-evaluate them on a noiseless, infinite-shot classical simulator. Overall, there is a general improvement in the inference error as shown in the bottom row of \Cref{fig:w_state_ibm_hardware_inference_optimization} when compared with the data collected from the noisy IBM hardware. Although the 1,000-shot case seems to be an exception with the error increasing slightly when evaluating the same parameters on the simulator, it is worth noting that the optimal parameters on noisy hardware might differ slightly from optimal on a simulator; hence, we view this run as mostly a statistical outlier. As further confirmed by additional experimental results in the appendix, while the data from the quantum hardware may seem noisy, the optimization over qubit measurements is indeed decreasing the inference error.

In \Cref{fig:w_state_ibm_optimal_inference} we show the optimal covariance and characteristic where the optimal matrix values are taken over all optimization steps. As the number of shots increases, so does the optimization's ability to resolve the correlation structure with greater accuracy. On the other hand, when the number of shots is small, statistical fluctuations can lead to stronger correlations than present, leading to false-positive correlations. For example, in the 10-shot case of \Cref{fig:w_state_ibm_optimal_inference} the zero terms of the covariance and characteristic matrices are optimized to be quite large, which would lead a researcher to infer that the two qubits are correlated. Also note that the IBM hardware can have significant errors on a given qubit for example, in the 10,000 shot characteristic matrix the top left matrix term remains close to zero despite taking more shots than other more successful trials. From experiments on the simulator, we know that mutual information generally is more difficult to optimize than covariance. Moreover, a noisy device might not be stable for the amount of time needed to acquire 10,000 shots. This instability leads to noise that inhibits the optimization's ability to make progress toward an optimum, particularly so for computing the mutual information. Therefore, unless one finds oneself working in the low-shot regime, covariance-based method is generally preferred as it exhibits more reliable convergence when presented with adequate resources.

\section{Conclusion}

In this work, we introduced protocols for distinguishing and inferring the topology of $n$-local quantum networks. The protocols construct matrices that encode entanglement structures, i.e., the characteristic matrix (cf.~\Cref{eq:char-mat,eq:qubit-char-mat}) and the covariance matrix. The entries of these matrices can be estimated using only local measurements, which allows for easy implementation on quantum hardware. Assuming sources prepare GHZ states, the characteristic matrix can uniquely determine the topology of a quantum network. Moreover, if one is capable of making qubit-wise measurements, the topology of the network can be inferred in polynomial time from both the characteristic matrix and the covariance matrix. Furthermore, the protocol is robust to noise and can be implemented on quantum hardware using quantum variational optimization. From experiments on both classical simulators and quantum hardware, we found that covariance-based methods are generally more stable during optimization, leading to more reliable discoveries of quantum entanglement. However, with limited shots, entropic protocols are more effective at avoiding false positives.

It is worth noting that the characteristic matrix cannot distinguish between quantum entanglement and shared randomness.
However, the characteristic matrix does indicate which qubits are correlated.
Thus, an entanglement witness can be tailored to the network's topology.
One approach might be to then test each source independently using an entanglement witness of choice~\cite{Terhal2002detecting,Guhne2009entanglement_detection}.

Future work should also focus on relaxing assumptions made in the manuscript for   even broader applications. For example, we assumed all measurement nodes are observed and measurements can be performed on all nodes. In reality, the known set of measurement nodes can merely be a subset of the entire network. Exploring the limits and extensions of our protocol for the case of partially-observed networks can be a fruitful future direction. On the other hand, the protocol for distinguishing network topology (cf.~\Cref{sec:distinguish-network}) relies on preparing the GHZ state, which is delicate and lacks robustness in practice. Extending the analysis to sources potentially distributing states with greater entanglement robustness, such as the W-class states~\cite{Vidal-1999a, Dur-2000a}, cluster states, or other partially entangled states would increase the applicability of this protocol. Another direction is to extend the protocol to infer network topology in more complex scenarios, for example, networks with intermediate processing nodes between source and measurement nodes, or communication between measurement devices/sources.

\section*{Code Availability}

Our data, numerical methods, and software tools are publicly available on GitHub \cite{qnetti_2023}.

\section*{Acknowledgements}

This material is based upon work supported by the U.S. Department of Energy, Office of Science, National Quantum Information Science Research Centers, and 
the Office of Advanced Scientific Computing Research, Accelerated Research for Quantum Computing program under contract number DE-AC02-06CH11357.  We acknowledge the use of IBM Quantum services for this work. The views expressed are those of the authors, and do not reflect the official policy or position of IBM or the IBM Quantum team.

\printbibliography

\vfil
\framebox{\parbox{.90\linewidth}{\scriptsize The submitted manuscript has been created by UChicago Argonne, LLC, Operator of Argonne National Laboratory (``Argonne''). Argonne, a U.S.\ Department of Energy Office of Science laboratory, is operated under Contract No.\ DE-AC02-06CH11357.  The U.S.\ Government retains for itself, and others acting on its behalf, a paid-up nonexclusive, irrevocable worldwide license in said article to reproduce, prepare derivative works, distribute copies to the public, and perform publicly and display publicly, by or on behalf of the Government.  The Department of Energy will provide public access to these results of federally sponsored research in accordance with the DOE Public Access Plan \url{http://energy.gov/downloads/doe-public-access-plan}.}}

\appendix 
\section{Proof of \Cref{thm:MMI-interpretation}}\label{sec:pf-mmi-interpretation}

In order to show \Cref{thm:MMI-interpretation}, we use the following established results reviewed in \Cref{thm:ent-lb} and \Cref{thm:ent-lb-2}. 

\begin{lemma}\label{thm:ent-lb}
Let $\sigma_n$ be an $n$-qubit shared uniform random bit, and let the classical probability distribution upon measuring $\sigma_n$ be $\mathbb P(\vec a)$. Then, the following inequality is true for any measurement basis,
\begin{align}
    H(\mathbb P(\vec a)) \geq 1
\end{align}
with equality occurring when measured in the computational basis.
\end{lemma}
\begin{proof}
    Recall that the entropy of the measurement of a state is minimized when measured in its eigenbasis~\cite{Nielsen2009}. An eigenbasis of $\sigma_n$ is the computational basis, which behaves like a classical coin flip upon measuring. Thus, the Shannon entropy is lower bounded by $1$.
\end{proof}

\begin{lemma}\label{thm:ent-lb-2}
Consider the distribution acquired through local measurements on a Bell state, that is 
\begin{align}
    \mathbb P(\vec a) = \bra \Phi \Pi_{a_1} \tensor \Pi_{a_2} \ket \Phi
\end{align}
for projective operators $\{\Pi_{a_1}\}$ and $\{\Pi_{a_2}\}$. Then, for any choice of $\Pi_{a_1}$ and $\Pi_{a_2}$, 
\begin{align}
    H(\mathbb P(\vec a)) \geq 1.
\end{align}
where the equality holds when measured in the computational basis.
\end{lemma}
\begin{proof}
    By subadditivity~\cite{Nielsen2009}, the Shannon entropy of measuring $\ket \Phi$ is at least the entropy of its marginal, which is $\sigma_1$. Since $\sigma_1$ is a classical coin flip, it has a Shannon entropy of $1$ and the entropy of $\ket \Phi$ is lower bounded by $1$. 
\end{proof}

Now, we present the proof for \Cref{thm:MMI-interpretation} below. 

\begin{proof}
Recall the definition of measured mutual information in \Cref{eq:meas-mut-info}. We must first determine the basis to measure in at each node, that is $\Pi^{A_i}_{\vec a_i}$ and $\Pi^{A_j}_{\vec a_j}$. 
Moreover, for any two nodes $A_i$ and $A_j$, the qubits received can be either maximally entangled or independent (maximally mixed). By \subassref{all}{A1}, the measurement basis does not influence the entropies at each node. Thus, we can achieve the lower bounds in \Crefrange{thm:ent-lb}{thm:ent-lb-2} by measuring in the computational basis.

Knowing the basis of choice, we proceed to understand the graph-theoretic properties of the mutual information. Let $N_s^{A_i}$ be the number of sources connected to device $A_i$. Then, we know that $H(\mathbb P(\vec a_i)) = N_s^{A_i}$ and similarly with $H(\mathbb P(\vec a_j))$ where $\mathbb P(\vec a_i)$ and $\mathbb P(\vec a_j)$ are probability distributions upon measuring at node $A_i$ and $A_j$ respectively. 

For the joint entropy $H(\mathbb P(\vec a_i, \vec a_j))$, partition $A_i \cup A_j$ into sets: 
\begin{align}
    S_1 &= \{(q_k, q_\ell) : q_k \in A_i, q_\ell \in A_j, \{q_k, q_\ell\} \subseteq \Lambda_i \text{ for some } i\} \\
    S_2 &= \{q_k : \forall q_\ell, \{q_k, q_\ell\} \not \in S_1\}
\end{align}
Each pair of qubits in $S_1$ are entangled and act jointly as a fair coin flip. On the other hand, each qubit in $S_2$ acts independently of one another also like a fair coin flip. Moreover, each element in $S_1$ and $S_2$ is independent of each other and Shannon entropies are additive. Thus, $H(\mathbb P(\vec a_i, \vec a_j)) = |S_1| + |S_2|$. Since the total number of qubits is $2|S_1| + |S_2| = H(\mathbb P(\vec a_i)) + H(\mathbb P(\vec a_j))$, the joint entropy can be expressed as
\begin{align}
    H(\mathbb P(\vec a_i, \vec a_j)) = H(\mathbb P(\vec a_i)) + H(\mathbb P(\vec a_j)) - |S_1|.
\end{align}
By definition of $S_1$, we find the measured mutual information to be
\begin{align}
    \mathcal I_m (A_i;A_j) = |S_1| = N_s^{A_i,A_j}. 
\end{align}
\end{proof}

\section{Proof of \Cref{thm:noiseless-classifier}}\label{sec:pf-noiseless-classifier}
\begin{proof}
For sufficiency, observe that if two networks have the same topology, then the number of sources connected to $A_i \ind 1$ is the same as $A_i \ind 2$, and the number of sources shared between $A_i \ind 1$ and $A_j \ind 1$ is the same as $A_i \ind 2$ and $A_j \ind 2$. By \Cref{thm:VNE-interpretation,thm:MMI-interpretation}, the von Neumann entropy of each node and measured mutual information of pairs of nodes are the same.

We now show necessity. Suppose the von Neumann entropy at each measurement node and the measured mutual information between any pair of measurement nodes for the two networks are identical.

First, we note that the two networks must have the same number of nodes; an immediate contradiction is reached otherwise. On the other hand, the two networks will also have the same number of links, which can be written alternatively using the von Neumann entropy as $N_\ell = \sum_{i} S(A_i)$
using \subassref{all}{A1}. Lastly, the number of sources must also be the same. Suppose not, and $N_s \ind 1 < N_s \ind 2$. Since the number of links present in either network is the same, there must be at least one link $\ell$ connected to the preparation node in $\mathcal N \ind 1$ that is not connected to the same preparation node in $\mathcal N \ind 2$. This means that the measurement device $A_i$ that was connected to $A_j$ via $\ell$ in $\mathcal N \ind 1$ must have lost the connection to $A_j$ in $\mathcal N \ind 2$. Thus, the $\mathcal I_m (A_i \ind 1;A_j \ind 1) > \mathcal I_m (A_i \ind 2; A_j \ind 2)$ and we've reached a contradiction.

Now, focus on the case where the two networks have the same number of sources, links, and nodes. Define a \textit{triplet} in a network to be a tuple of three elements, $(A_i, \Lambda_k, A_j)$, such that $A_i$ and $A_j$ share $\Lambda_k$. For two $n$-local quantum networks, we construct the map $\xi$ such that:
\begin{enumerate}[label=(\roman*)]
    \item $\xi ( (A_i \ind 1, \Lambda_k \ind 1, A_j \ind 1) ) = (A_i \ind 2, \Lambda_k \ind 2, A_j \ind 2)$ where $(A_i \ind 1, \Lambda_k \ind 1, A_j \ind 1)$ is in $\mathcal N \ind 1$ and $(A_i \ind 2, \Lambda_k \ind 2, A_j \ind 2)$ is in $\mathcal N \ind 2$, and
    \item performing the map $\xi$ on all triplets $t_n$ in $\mathcal N \ind 1$ yields $\mathcal N \ind 2$. 
\end{enumerate}
Think of $\xi$ as the map that ``moves'' $\mathcal N \ind 1$ to $\mathcal N \ind 2$. In particular, when a triplet is present in both networks, we take $\xi$ to be the identity; when a triplet is only present in one network, then $\xi$ moves the corresponding triplet into a new location. This restriction is important since this removes the case that $\xi$ cyclically moves edges, e.g., $(1,1,2) \mt (2,1,3)$, $(2,1,3) \mt (3,1,1)$, and $(3,1,1) \mt (1,1,2)$.

This map can induce relabeling of indices. Let $\phi: [N_s] \to [N_s]$ be the map that relabels the source indices, and $\varphi: [N_m] \to [N_m]$ be the one for measurement nodes. Then, we can can define $\phi$ and $\varphi$ using $\xi$: if $\xi(A_i \ind 1, \Lambda_k \ind 1, A_j \ind 1) = (A_i \ind 2, \Lambda_k \ind 2, A_j \ind 2)$, then define
\begin{align}
    \phi(\Lambda_k \ind 1) = \Lambda_k \ind 2, ~~ \varphi(A_i \ind 1) = A_i \ind 2, ~~ \varphi(A_j \ind 1) = A_j \ind 2.
\end{align}
If $\phi$ and $\varphi$ are well-defined, then they are immediately bijective because we're comparing networks of equal sizes---nodes have to be mapped to (onto) and two nodes cannot be ``squished'' and become indistinguishable (one-to-one). Therefore, we will show that if $\xi$ preserves the characteristic matrix, then $\phi$ and $\varphi$ are well-defined and the two networks have the same topology. We do both by contradiction.

If $\phi$ is not well-defined, then there are $i,j,m,n$ ($i \neq m$) such that
\begin{align}
    \xi(A_i \ind 1, \Lambda_k \ind 1, A_j \ind 1) = (A_i \ind 2, \Lambda_k \ind 2, A_j \ind 2)    
\end{align}
but 
\begin{align}
    \xi(A_m \ind 1, \Lambda_k \ind n, A_n \ind 1) = (A_m \ind 2, \Lambda_\ell \ind 2, A_n \ind 2)    
\end{align}
for some $\ell \neq k$. However, this means that connections between $A_i \ind 1$ and $A_m \ind 1$ through $\Lambda_k \ind 1$ are moved away upon applying $\xi$ and 
\begin{align}
    \mathcal I_m (A_i \ind 2, A_m \ind 2) \leq \mathcal I_m (A_i \ind 1, A_m \ind 1) - 1,
\end{align}
which contradicts the assumption that $\xi$ preserves the measured mutual information.

On the other hand, if $\varphi$ is not well-defined, then there are $j,k,\ell,m$ ($k \neq \ell$) such that 
\begin{align}
    \xi(A_i \ind 1, \Lambda_k \ind 1, A_j \ind 1) = (A_i \ind 2, \Lambda_k \ind 2, A_j \ind 2)    
\end{align}
but 
\begin{align}
    \xi(A_i \ind 1, \Lambda_\ell \ind 1, A_n \ind 1) = (A_m \ind 2, \Lambda_\ell \ind 2, A_n \ind 2)    
\end{align}
for some $m \neq i$. However, this means that the connection to $A_i \ind 1$ from $\Lambda_\ell \ind 1$ is moved away so 
\begin{align}
    S(A_i \ind 1) \leq S(A_i \ind 2) - 1
\end{align}
which contradicts the assumption that $\xi$ preserves the Shannon entropy. Thus, if two networks have the same characteristic matrix, then the two networks have the same topology.

\end{proof}

\section{Mutual information and covariance of channel noise} \label{sec:noisy-channel-calc}
Let $\rho$ be the density matrix of the Bell state. we hope to calculate the mutual information and covariance between the two qubits after applying depolarizing and amplitude-damping noise independently on each qubit. We approach the calculation the same way. We first write down the density operator through the respective quantum channels, $\mathcal E(\rho)$, via the Kraus operator formalism. Then, we will find the classical distribution induced by measurement in the basis that maximizes the correlation quantity of interest.

\paragraph{Depolarizing noise} Quantum channel that applies qubit-depolarizing noise can be written as the following summation:
\begin{align}
    \mathcal E(\rho) = \sum_{i, j} (K_i \tensor K_j) \rho (K_i \tensor K_j)^\dagger \label{eq:kraus-op}
\end{align}
where $\{K_i\}$ are the Kraus operators of a single-qubit depolarizing channel
\begin{align}
K_1 = \sqrt{\frac{3\gamma}{4}} \mathbb I_2,~~ &K_2 = \sqrt{\frac{\gamma}{4}} \sigma_x, \nonumber \\
~~ &K_3 = \sqrt{\frac{\gamma}{4}} \sigma_y, ~~ K_4 = \sqrt{\frac{\gamma}{4}} \sigma_z.
\end{align}
Expanding the summation, the quantum state $\mathcal E(\rho)$ can be compactly written as
\begin{align}
    \mathcal E(\rho) = \frac{\gamma(2-\gamma)}{4} \mathbb I_4 + (1-\gamma)^2 \rho.
\end{align}
Measuring the above state in the $\sigma_z$ basis yields the following probability distribution:
\begin{align}
    \mathbb P(x_1x_2) = 
    \begin{cases}
        (1+(1-\gamma)^2)/4 & \IF~x_1 = x_2 \\
        \gamma(2-\gamma)/4 & \OW.
    \end{cases}
\end{align}
The remainder calculation follows the definition of measured mutual information and covariance, which will lead to \Crefrange{eq:depolarizing-mi}{eq:depolarizing-cov}.

\paragraph{Amplitude damping} We follow a similar procedure in the case of applying amplitude-damping noise. Again, we can write the noisy quantum state similar to \Cref{eq:kraus-op} but with Kraus operators
\begin{align}
    K_1 = \begin{pmatrix} 1 & 0 \\ 0 & 1-\gamma \end{pmatrix}, ~~K_2 = \begin{pmatrix} 0 & \sqrt{\gamma} \\ 0 & 0 \end{pmatrix}.
\end{align}
Expanding the summation yields the following density matrix:
\begin{align}
\mathcal E (\rho_\Phi) = \frac{1}{2} \begin{pmatrix} 1 + \gamma^2 & 0 & 0 & 1-\gamma \\ 0 & \gamma(1-\gamma) & 0 & 0 \\ 0 & 0 & \gamma(1-\gamma) & 0 \\ 1-\gamma & 0 & 0 & (1-\gamma)^2 \end{pmatrix},
\end{align}
and measuring in the $\sigma_x$ basis---the basis that maximizes the observed correlation---is equivalent to measuring the state
\begin{align}
\frac{1}{4} \begin{pmatrix} 2-\gamma & \gamma & \gamma & 2\gamma^2 - 3\gamma + 2 \\ \gamma & \gamma & 2\gamma^2-\gamma & \gamma \\ \gamma & 2\gamma^2-\gamma & \gamma & \gamma \\ 2\gamma^2 - 3\gamma + 2 & \gamma & \gamma & 2 - \gamma \end{pmatrix} 
\end{align}
in the computational basis, that is, the induced classical distribution is encoded in the diagonal. Following definitions of mutual information and covariances respectively again yields \Crefrange{eq:ampdamp-mi}{eq:ampdamp-cov}.

\section{Inference of 5-Qubit Zero State and 5-Qubit GHZ State IBM Hardware}\label{appendix:ibm_inference_data_5-qubit_zero_and_ghz}

\begin{figure}[b!]
    \centering
    \includegraphics[width=0.95\textwidth]{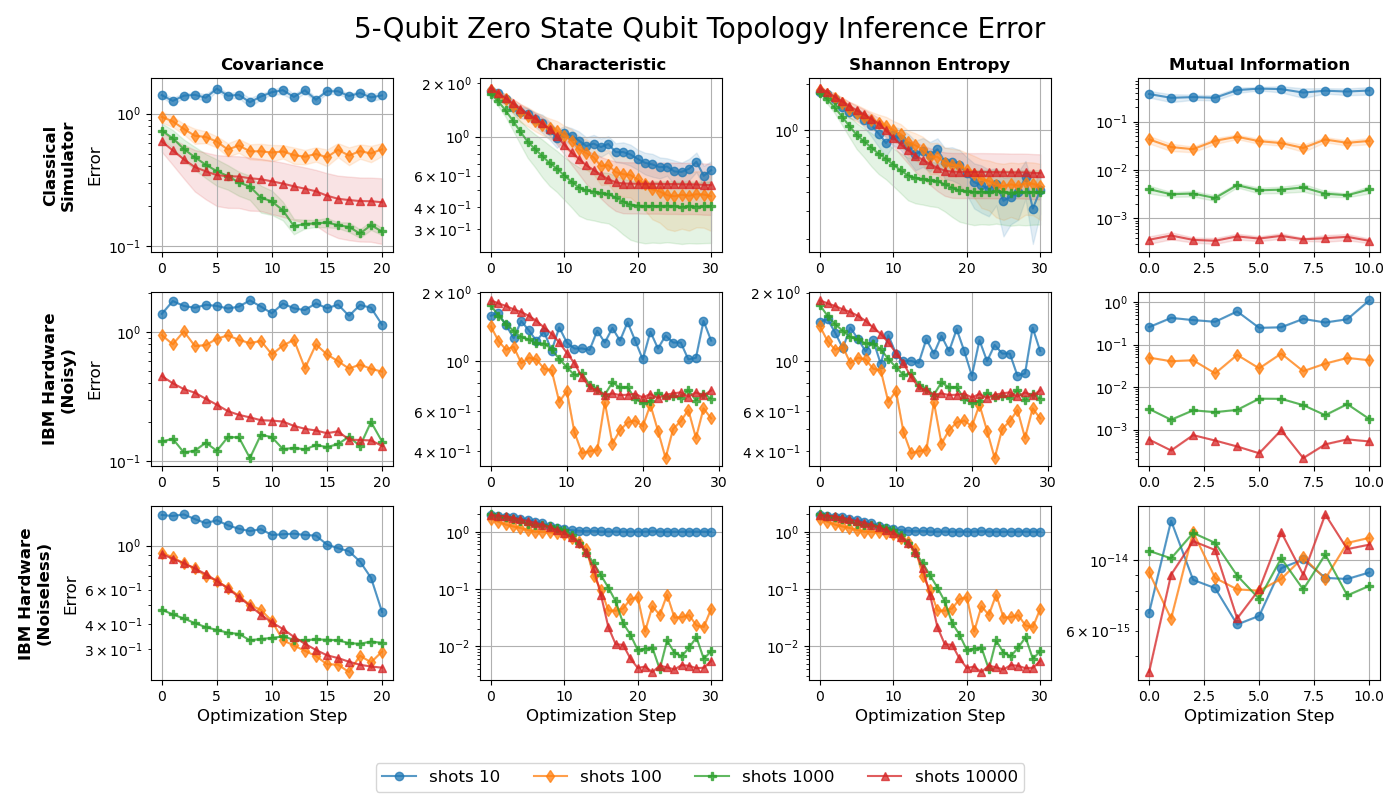}
    \caption{Variational quantum optimization of the covariance matrices and characteristic matrices for a 5-qubit zero state prepared on IBM Hardware. In each plot the blue circles show the 10-shot optimization, the orange diamonds show the 100-shot optimization, green plus signs show the 1,000-shot optimization, and red triangles show the 10,000-shot optimization. The $x$-axis shows the optimization step while the $y$-axis shows the inference error calculated as the Euclidean distance between the ideal covariance/characteristic matrix and the matrix in each optimization step. The first row shows the optimization data averaged over 10 runs on a finite-shot noiseless classical simulator, the second row shows the optimization data collected from the \texttt{ibmq\_belem} quantum computer, and the third row shows the data collected when the settings from the noisy IBM hardware optimization are reevaluated on a noiseless classical simulator. From left to right each column shows the covariance matrix optimization, the characteristic matrix optimization, the Shannon entropy optimization, and the classical mutual information optimization.}
    \label{fig:5-qubit_zero_state_ibm_hardware_inference_optimization}
\end{figure}

\begin{figure}[t!]
    \centering
    \includegraphics[width=0.95\textwidth]{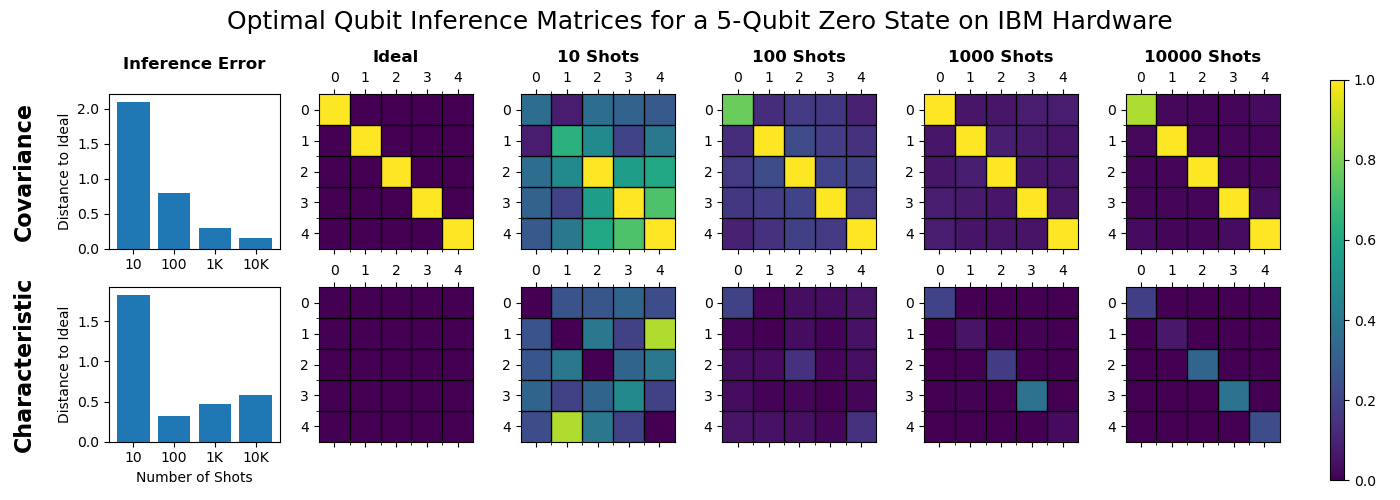}
    \caption{In the top row, we plot the maximal value achieved for each term of the covariance matrix across all optimization steps. In the bottom row, we plot the maximal mutual information of each qubit pair (off-diagonals) and minimal Shannon entropy of each qubit (diagonal) across all optimization steps. The bar graph plots the Euclidean distance between the ideal and inferred matrices for each distinct number of shots.}
    \label{fig:5-qubit_zero_state_ibm_optimal_inference}
\end{figure}

\begin{figure}[b!]
    \centering
    \includegraphics[width=0.95\textwidth]{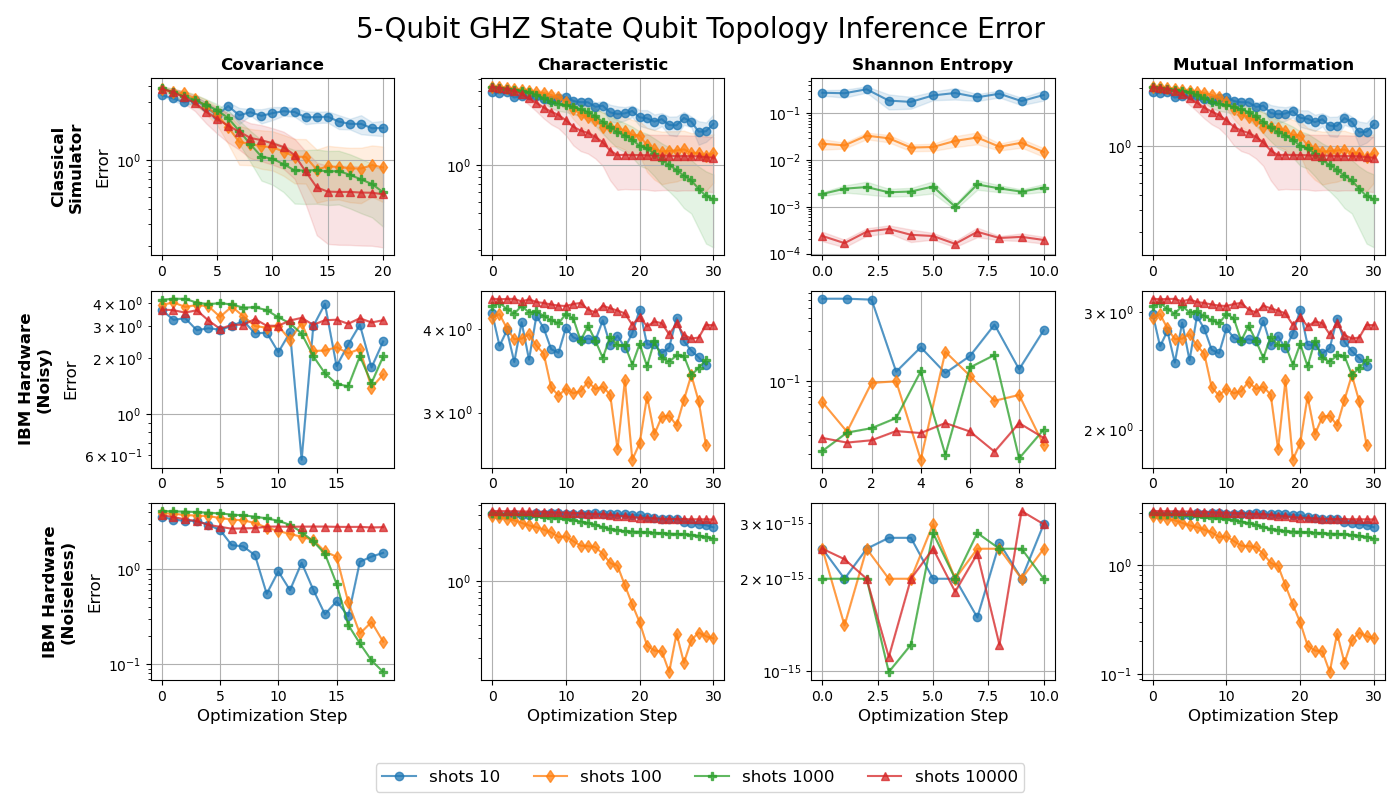}
    \caption{Variational quantum optimization of the covariance and characteristic matrices for a 5-qubit GHZ state prepared on IBM Hardware. See \Cref{fig:5-qubit_zero_state_ibm_hardware_inference_optimization} for descriptions of individual plots.}
    \label{fig:5-qubit_ghz_state_ibm_hardware_inference_optimization}
\end{figure}

In this section, we present the remainder of the results from our numerical experiments on IBM quantum hardware. Namely, we show the inference error when our variational scheme is applied to a 5-qubit zero state and a 5-qubit GHZ state. The former isolates the optimization of the von Neumann entropy, while the latter isolates the optimization of covariance and mutual information. We follow the same approach used for the inference of the $\ket{W}\otimes \ket{\Phi}$ as discussed in the main body of the paper.

In \Cref{fig:5-qubit_zero_state_ibm_hardware_inference_optimization} and \Cref{fig:5-qubit_zero_state_ibm_optimal_inference} we plot the network inference for a 5-qubit zero state preparation $\ket{\psi} = \ket{00000}$. Since $\ket{\psi}$ is separable, for any two-qubit measurement $\Pi^{q_i}\otimes \Pi^{q_j}$, the mutual information $I(q_i; q_j)=0$ and  the covariance  is $\text{Cov}(q_1,q_2) = 0$. As a result, randomly initialized measurements are optimal as shown in the simulated results in the top-right plot of \Cref{fig:5-qubit_zero_state_ibm_hardware_inference_optimization}. On the other hand, the variance of any qubit is $\text{Var}(q_i)\in [0,1]$ and the Shannon entropy $H(\mathbb{P}(a_i)) \in [0,1]$ where the maximal variance and minimal Shannon entropy are both achieved only when the state $\op{0}{0}$ is measured in the computational basis. Therefore, randomly initialized measurements must be optimized to achieve the optimal covariance matrix and characteristic matrix, which are the $5\times5$ identity and zero matrix, respectively.

As the number of shots increases we expect to see the network inference error decrease. This trend can easily be seen in the classical simulator data in the top row of \Cref{fig:5-qubit_zero_state_ibm_hardware_inference_optimization}. First, note in the rightmost column that a 10x increase in the number of shots corresponds to roughly a 10x decrease in the amount of error. In the covariance, characteristic, and Shannon entropy optimizations, we find a similar pattern in which more shots leads to less error, however, the separation is not as profound as in the mutual information case because the covariance and Shannon entropy optimizations over measurements have some error. This error is reflected in the standard error bar on the classical simulator optimization, which plots the average optimization over 10 independent optimization runs.  When we run the same optimizations on the IBM hardware, we find a similar trend where increasing the number of shots decreases the error.

In \Cref{fig:5-qubit_zero_state_ibm_optimal_inference}, we visualize the network inference performance by taking the optimal matrix elements for the covariance and characteristic matrices across the whole optimization on IBM hardware. For the covariance-based inference case, we find that the inference error decreases as the number of shots increases. This trend does not hold for the characteristic matrix, which shows the error increases when more than 100 shots are considered. We largely attribute this feature to the noise on the IBM hardware, which is quite dynamic and can vary considerably over the course of an optimization.

\begin{figure}[t!]
    \centering
    \includegraphics[width=0.95\textwidth]{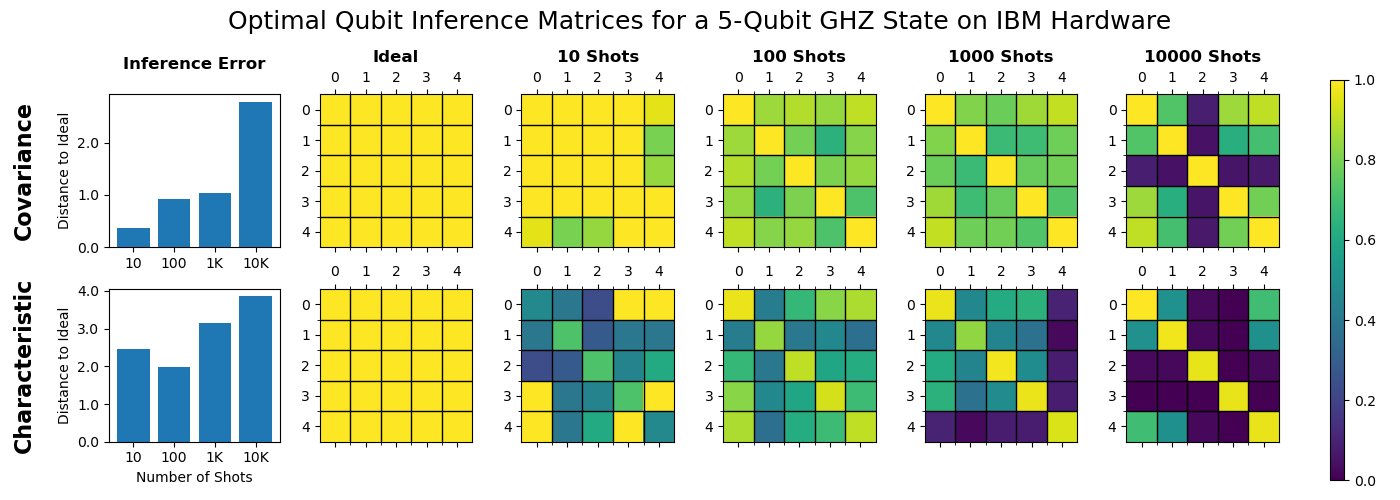}
    \caption{In the top row, we plot the maximal value achieved for each term of the covariance matrix across all optimization steps. In the bottom row, we plot the maximal mutual information of each qubit pair (off-diagonals) and minimal Shannon entropy of each qubit (diagonal) across all optimization steps. The bar graph plots the Euclidean distance between the ideal and inferred matrices for each distinct number of shots.}
    \label{fig:5-qubit_ghz_state_ibm_optimal_inference}
\end{figure}

In \Cref{fig:5-qubit_ghz_state_ibm_hardware_inference_optimization} and \Cref{fig:5-qubit_ghz_state_ibm_optimal_inference} we plot the network inference for a 5-qubit GHZ state preparation $\ket{\psi} = \frac{1}{\sqrt{2}}(\ket{00000}+\ket{11111})$. For any two-qubit measurement $\Pi^{q_i}\otimes \Pi^{q_j}$, the mutual information $I(q_i; q_j)\in[0,1]$ and covariance  $\text{Cov}(q_i,q_j) \in[0,1]$ where these values are maximal when $\Pi^{q_i} = \Pi^{q_j}$ are the same measurement. As a result, randomly initialized measurements must be optimized to obtain their maximal values. On the other hand, the variance of any qubit is $\text{Var}(q_i)=1$ and the Shannon entropy $H(\mathbb{P}(a_i)) =1$ for any choice of measurement, hence randomly initialized measurements are optimal. Note that in the Shannon entropy plots of \Cref{fig:5-qubit_ghz_state_ibm_hardware_inference_optimization}. The ideal qubit covariance and characteristic matrices are both $5\times5$ matrices of ones. 

As the number of shots increases we expect to see the network inference error decrease. This trend can easily be seen in the classical simulator data in the top row of \Cref{fig:5-qubit_ghz_state_ibm_hardware_inference_optimization}. For instance, note that for the Shannon entropy optimization, a 10x increase in the number of shots corresponds to roughly a 10x decrease in the inference error, which is a result of the Shannon entropy being constant for all measurements. For the other optimizations, the 10,000-shot case does not necessarily perform the best, but this can largely be attributed to finding local optima during optimization, causing the mean inference error to level off with wide error bars.  When we run the same optimizations on the IBM hardware, we find that more shots do not necessarily decrease the inference error. We attribute this feature to the fact the IBM quantum hardware is noisy. Furthermore, the noise is not necessarily constant over the course of an optimization run.  Thus, we suspect that considering too many shots allowed the noise to drift during optimization skewing the optimization. Furthermore, the numerical experiments were run serially from fewer shots to more shots, meaning that the performance of the IBM device might have also deteriorated considerably after the completion of the 100-shot experiment.

The noise in the IBM hardware is visualized in \Cref{fig:5-qubit_ghz_state_ibm_optimal_inference}. For both the covariance and characteristic matrices, we find the inference error to increase as the number of shots increases. First, this trend results from the fact that we take the maximal covariance and mutual information for each qubit pair, thus, in the 10-shot case statistical fluctuations lead to a larger covariance and mutual information, implying less error. Second, we observe in the 1,000 and 10,000-shot cases that the inference scheme fails for certain qubits. In these cases, either local optima are being found, or the IBM hardware is failing to produce entangled states. From our classical simulator data, we find that characteristic matrix optimizations are particularly susceptible to finding local optima, meaning that the results in \Cref{fig:5-qubit_ghz_state_ibm_optimal_inference} could simply be the optimization finding local optima. In practice, these errors can be mitigated by running the optimization repeatedly from different initial settings.

\end{document}